\documentclass{article}
\usepackage[left=1in, right=1in, top=1in, bottom=1in]{geometry}
\usepackage{enumitem}
\usepackage{authblk}
\usepackage{graphicx}	
\usepackage{amsmath, amsthm, amssymb}
\usepackage{xspace}
\usepackage{comment}
\usepackage{hyperref}
\usepackage{xcolor}

\newcommand{\algprobm}[1]{\textsc{#1}\xspace}

\theoremstyle{plain}
\newtheorem{theorem}{Theorem}[section]

\newtheorem{corollary}[theorem]{Corollary}
\newtheorem{lemma}[theorem]{Lemma}

\theoremstyle{definition}
\newtheorem{definition}[theorem]{Definition}
\newtheorem{remark}[theorem]{Remark}
\newtheorem{question}[theorem]{Question}\newtheorem{algorithm}[theorem]{Algorithm}

\newcommand{\Lem}[1]{Lemma~\ref{#1}\xspace}
\newcommand{\Cor}[1]{Corollary~\ref{#1}\xspace}

\newcommand{\Thm}[1]{Theorem~\ref{#1}\xspace}

\DeclareMathOperator{\rw}{rw}

\DeclareMathOperator{\rk}{rk}

\title{Canonizing Graphs of Bounded Rank-Width in Parallel via Weisfeiler--Leman\footnote{This work was completed in part at the 2022 Graduate Research Workshop in Combinatorics, which was supported in part by NSF grant \#1953985 and a generous award from the Combinatorics Foundation. ML was partially supported by J. A. Grochow's NSF award CISE-2047756 and the University of Colorado Boulder, Department of Computer Science Summer Research Fellowship. We are grateful to the anonymous referees for their helpful feedback.} }

\author[1]{Michael Levet}
\author[2]{Puck Rombach}
\author[3]{Nicholas Sieger}
\affil[1]{Department of Computer Science, College of Charleston}
\affil[2]{Department of Mathematics and Statistics, University of Vermont}
\affil[3]{Department of Mathematics, University of California San Diego}

\begin{document}
\maketitle

\begin{abstract}
In this paper, we show that computing canonical labelings of graphs of bounded rank-width is in $\textsf{TC}^{2}$. Our approach builds on the framework of Köbler \& Verbitsky (CSR 2008), who established the analogous result for graphs of bounded treewidth. Here, we use the framework of Grohe \& Neuen (\textit{ACM Trans. Comput. Log.}, 2023) to enumerate separators via \textit{split-pairs} and \textit{flip functions}. In order to control the depth of our circuit, we leverage the fact that any graph of rank-width $k$ admits a rank decomposition of width $\leq 2k$ and height $O(\log n)$ (Courcelle \& Kant\'e, WG 2007). This allows us to utilize an idea from Wagner (CSR 2011) of tracking the depth of the recursion in our computation.

Furthermore, after splitting the graph into connected components, it is necessary to decide isomorphism of said components in $\textsf{TC}^{1}$. To this end, we extend the work of Grohe \& Neuen (\textit{ibid}.) to show that the $(6k+3)$-dimensional Weisfeiler--Leman (WL) algorithm can identify graphs of rank-width $k$ using only $O(\log n)$ rounds. As a consequence, we obtain that graphs of bounded rank-width are identified by $\textsf{FO} + \textsf{C}$ formulas with $6k+4$ variables and quantifier depth $O(\log n)$. 

Prior to this paper, isomorphism testing for graphs of bounded rank-width was not known to be in $\textsf{NC}$.
\end{abstract}

\thispagestyle{empty}

\newpage

\setcounter{page}{1}

\section{Introduction}
\label{sec:introduction}

The \textsc{Graph Isomorphism} problem (\textsc{GI}) takes as input two graphs $G$ and $H$, and asks if there exists an isomorphism $\varphi : V(G) \to V(H)$. $\textsc{GI}$ is in particular conjectured to be $\textsf{NP}$-intermediate; that is, belonging to $\textsf{NP}$ but neither in $\textsf{P}$ nor $\textsf{NP}$-complete \cite{Ladner}. Algorithmically, the best known upper-bound is $n^{\Theta(\log^{2} n)}$, due to Babai \cite{BabaiGraphIso}. It remains open as to whether $\textsc{GI}$ belongs to $\textsf{P}$. There is considerable evidence suggesting that $\textsc{GI}$ is not $\textsf{NP}$-complete \cite{Schoning, BuhrmanHomer, ETH, BabaiGraphIso, GILowPP, ArvindKurur, MATHON1979131}. In a precise sense, \algprobm{GI} sits between linear and multilinear algebra. Recent works \cite{FutornyGrochowSergeichukTensorWild, GrochowQiaoTensors} have established that for any field $\mathbb{F}$, \textsc{GI} belongs to $\mathbb{F}$-$\textsf{Tensor Isomorphism}$ ($\textsf{TI}_{\mathbb{F}}$). When $\mathbb{F}$ is finite, $\textsf{TI}_{\mathbb{F}} \subseteq \textsf{NP} \cap \textsf{coAM}$.\footnote{We refer the reader to \cite[Remark~3.4]{GrochowQiaoTensors} for discussion on $\textsf{TI}_{\mathbb{F}}$ when $\mathbb{F}$ is infinite.} In contrast, the best known lower-bound for \textsc{GI} is $\textsf{DET}$ \cite{Toran}, which contains $\textsf{NL}$ and is a subclass of $\textsf{TC}^{1}$. It is thus natural to inquire as to families of graphs where isomorphism testing is decidable in complexity classes contained within $\textsf{DET}$.

There has been considerable work on efficient algorithms for special families of graphs. Sparse graphs, in particular, have received considerable attention from the perspective of polynomial-time computation, including notably planar graphs \cite{HopcroftWong, KieferPonomarenkoSchweitzer, GroheKieferPlanar, DattaLimayeNimbhorkarPrajaktaThieraufWagner}, graphs of bounded genus, \cite{FilottiMayer, MillerGIBoundedGenus, GroheBook, GroheKieferBoundedGenus}, graphs of bounded treewidth \cite{BodlaenderTreewidth}, all classes excluding a fixed graph as a minor \cite{PonomarenkoMinor, GroheForbiddenMinor} or even a topological subgraph \cite{GroheMarxTopologicalSubgraph}, and graphs of bounded degree \cite{LUKS198242, BabaiKantorLuksBoundedDegree,  BabaiLuksCanonicalLabeling, FurerSchnyderSpecker, GroheNeuenSchweitzerBoundedDegree}. Less is known about dense graphs. Polynomial-time algorithms are known for graphs of bounded eigenvalue multiplicity \cite{BabaiGrigoryevMount}, and certain hereditary graph classes including intersection graphs \cite{CurtisLinMcconnellNussbaum}, graphs excluding specific induced subgraphs \cite{KRATSCH2017240, SchweitzerDichotomy}.

The story is similar in the setting of efficient parallel ($\textsf{NC}$) isomorphism tests. There has been considerable work on $\textsf{NC}$ algorithms for planar graphs (see the references in \cite{DattaLimayeNimbhorkarPrajaktaThieraufWagner}) and graphs of bounded treewidth \cite{GroheVerbitsky, WagnerBoundedTreewidth, DasToranWagner}, culminating in $\textsf{L}$-completeness results for both families (see \cite{DattaLimayeNimbhorkarPrajaktaThieraufWagner, ThieraufWagnerPlanarUL} for planar graphs, and and \cite{ElberfeldSchweitzer} for graphs of bounded treewidth). Isomorphism testing for graphs of bounded genus is also known to be $\textsf{L}$-complete \cite{ElberfeldKawarabayashiGenus}. We now turn to dense graphs. An $\textsf{NC}$ isomorphism test is known for graphs of bounded eigenvalue multiplicity \cite{BabaiLuksSeress}. The isomorphism problems for certain hereditary graph classes, including interval graphs \cite{IntervalGILogspace} and Helly ciruclar-arc graphs \cite{KOBLER2016266} are both $\textsf{L}$-complete.

The $k$-dimensional Weisfeiler--Leman algorithm ($k$-WL) serves as a key combinatorial tool in \textsc{GI}. It works by iteratively coloring $k$-tuples of vertices in an isomorphism-invariant manner. On its own, Weisfeiler--Leman serves as an efficient polynomial-time isomorphism test for several families of graphs, including planar graphs~\cite{KieferPonomarenkoSchweitzer, GroheKieferPlanar}, graphs of bounded genus~\cite{GroheBoundedGenus, GroheKieferBoundedGenus}, and graphs for which a specified minor $H$ is forbidden~\cite{GroheForbiddenMinor}. We also note that $1$-WL identifies almost all graphs~\cite{BabaiKucera, BabaiErdosSelkow} and $2$-WL identifies almost all regular graphs~\cite{BollobasRegular, KuceraRegular}. In the case of graphs of bounded treewidth~\cite{GroheVerbitsky} and planar graphs~\cite{GroheVerbitsky, VerbitskyPlanar,  GroheKieferPlanar}, Weisfeiler--Leman serves even as an $\textsf{NC}$ isomorphism test. Despite the success of WL as an isomorphism test, it is insufficient to place \textsc{GI} into $\textsf{P}$~\cite{CFI, NeuenSchweitzerIR}. Nonetheless, WL remains an active area of research. For instance, Babai's quasipolynomial-time algorithm~\cite{BabaiGraphIso} combines $O(\log n)$-WL with group-theoretic techniques.

Graphs of bounded rank-width have only recently received attention from the perspective of isomorphism testing. Grohe \& Schweitzer \cite{GroheSchweitzerRankWidth} gave the first polynomial-time isomorphism test for graphs of bounded rank-width. In particular, their isomorphism test ran in time $n^{f(k)}$, where $f(k)$ was a non-elementary function of the rank-width $k$. Subsequently, Grohe \& Neuen \cite{grohe2019canonisation} showed that graphs of rank-width $k$ have Weisfeiler--Leman dimension $\leq 3k+5$, which yields an $O(n^{3k+6} \log n)$-time isomorphism test and also the first polynomial-time canonical labeling procedure for this family. In particular, it is open as to whether graphs of bounded rank-width admit $\textsf{NC}$ or $\textsf{FPT}$ isomorphism tests. This is in contrast to graphs of bounded treewidth, where $\textsf{NC}$ \cite{GroheVerbitsky, KoblerVerbitsky, WagnerBoundedTreewidth, DasToranWagner, ElberfeldSchweitzer} and  $\textsf{FPT}$ \cite{LokshtanovTreewidth, GroheNeuenSchweitzerWiebkingTreewidth} isomorphism tests are well-known.

Closely related to \algprobm{Graph Isomorphism} is \algprobm{Graph Canonization}, which for a class $\mathcal{C}$ of graphs, asks for a function $F : \mathcal{C} \to \mathcal{C}$ such that for all $X, Y \in \mathcal{C}$, $X \cong F(X)$ and $X \cong Y \iff F(X) = F(Y)$. \algprobm{Graph Isomorphism} reduces to \algprobm{Graph Canonization}, and the converse remains open. Nonetheless, efficient canonization procedures have often followed efficient isomorphism tests, usually with non-trivial work-- see e.g., \cite{ImmermanLander1990, grohe2019canonisation, KoblerVerbitsky, WagnerBoundedTreewidth, ElberfeldSchweitzer, BabaiQuasipolynomialCanonization}.

\noindent \\ \textbf{Main Results.} In this paper, we investigate the parallel and descriptive complexities of identifying and canonizing graphs of bounded rank-width, using the Weisfeiler--Leman algorithm. 

\begin{theorem} \label{thm:MainCanonicalForms}
Let $G$ be a graph on $n$ vertices, of rank-width $k$. We can compute a canonical labeling for $G$ using a $\textsf{TC}$ circuit of depth $O(\log^{2} n)$ and size $n^{O(16^k)}$. 
\end{theorem}

For $k \in O(1)$, Theorem~\ref{thm:MainCanonicalForms} yields a $\textsf{TC}^{2}$ upper bound. 

\noindent Our approach in proving Theorem~\ref{thm:MainCanonicalForms} was inspired by the previous work of K\"obler \& Verbitsky \cite{KoblerVerbitsky}, who established the analogous result for graphs of bounded treewidth. K\"obler \& Verbitsky crucially utilized the fact that graphs of treewidth $k$ admit \emph{balanced} separators of size $k+1$, where removing such a separator leaves connected components each of size $\leq n/2$. This ensures that the height of their recursion tree is $O(\log n)$. 

For graphs of bounded rank-width, we are unable to identify such separators. Instead, we leverage the framework of Grohe \& Neuen to descend along a rank decomposition, producing a canonical labeling along the way. To ensure that our choices are canonical, we utilize the Weisfeiler--Leman algorithm. As a first step, we will establish the following:

\begin{theorem} \label{thm:MainParallel1}
The $(6k+3)$-dimensional Weisfeiler--Leman algorithm identifies graphs of rank-width $k$ in $O(\log n)$ rounds.
\end{theorem}

Combining \Thm{thm:MainParallel1} with the parallel WL implementation of Grohe \& Verbitsky \cite{GroheVerbitsky}, we obtain the first $\textsf{NC}$ bound for isomorphism testing of graphs of bounded rank-width. This is a crucial ingredient in obtaining the $\textsf{TC}^{2}$ bound for \Thm{thm:MainCanonicalForms}.

\begin{corollary}
Let $G$ be a graph of rank-width $O(1)$, and let $H$ be arbitrary. We can decide isomorphism between $G$ and $H$ in $\textsf{TC}^{1}$. 
\end{corollary}

Furthermore, in light of the close connections between Weisfeiler--Leman and $\textsf{FO} + \textsf{C}$ \cite{ImmermanLander1990, CFI}, we obtain the following corollary. Let $\mathcal{C}_{m,r}$ denote the $m$-variable fragment of $\textsf{FO} + \textsf{C}$ where the formulas have quantifier depth at most $r$ (see Sec.~\ref{sec:Logics}).

\begin{corollary}
For every graph $G$ of rank-width at most $k$, there is a sentence $\varphi_{G} \in \mathcal{C}_{6k+4, O(\log n)}$ that characterizes $G$ up to isomorphism. That is, whenever $H \not \cong G$, we have that $G \models \varphi_{G}$ and $H \not \models \varphi_{G}$.
\end{corollary}

We will discuss shortly the proof technique for \Thm{thm:MainParallel1}. Let us now discuss how we will utilize \Thm{thm:MainParallel1} to establish \Thm{thm:MainCanonicalForms}. As $(6k+3)$-WL identifies all graphs of rank-width $\leq k$ in $O(\log n)$ rounds, $(10k+3)$-WL identifies the orbits of sequences of vertices of length $\leq 4k$ (see Lemma~\ref{lem:Orbits}). By applying $(10k+3)$-WL for $O(\log n)$ rounds at each recursive call to our canonization procedure, we thereby give canonical labelings to the various parallel choices considered by the algorithm.  While there exists a suitable rank decomposition of height $O(\log n)$ \cite{CourcelleKante2007}, it is open whether such a decomposition can be computed in $\textsf{NC}$ \cite{DasAnirbanMuraliReddy}. Instead of explicitly constructing a rank decomposition, we instead track the depth of our recursion tree. By leveraging the framework of Grohe \& Neuen \cite{grohe2019canonisation}, we are able to show that one of our parallel computations witnesses the balanced rank decomposition of Courcelle \& Kant\'e \cite{CourcelleKante2007}.

We will now outline the proof strategy for \Thm{thm:MainParallel1}. Our work follows closely the strategy of Grohe \& Neuen \cite{grohe2019canonisation}. We again combine the balanced rank decomposition of Courcelle \& Kant\'e \cite{CourcelleKante2007} with a careful analysis of the pebbling strategy of Grohe \& Neuen \cite{grohe2019canonisation}. In parts of their argument, Grohe \& Neuen utilize (an analysis of) the stable coloring of $1$-WL. 

For a graph $G$, Grohe \& Neuen \cite[Section~3]{grohe2019canonisation} construct an auxiliary graph that they call the \textit{flipped graph}, whose construction depends on a specified set of vertices called a \textit{split pair} and a coloring of the vertices. While the flipped graph is compatible with any vertex coloring, Grohe \& Neuen \cite[Lemma~3.6]{grohe2019canonisation} crucially utilize the \textit{stable coloring} of $1$-WL to show that WL can detect which edges are present in the flipped graph. Even though we allow for higher-dimensional WL, the restriction of $O(\log n)$ rounds creates a technical difficulty in determining whether WL detects the key properties needed to adapt \cite[Lemma~3.6]{grohe2019canonisation}.

To resolve this issue, we consider a \textit{different} notion of flipped graph-- namely, a vertex colored variant introduced in \cite[Section~5]{grohe2019canonisation}. In this second definition, edges of the flipped graph no longer depend on the coloring and merely depend on the split pair. Grohe \& Neuen established \cite[Lemma~5.6]{grohe2019canonisation}, which is analogous to their Lemma~3.6. The proof of their Lemma~5.6 depends only on the structure of the graph and \textit{not} the vertex colorings. In particular, Weisfeiler--Leman can take advantage of \cite[Lemma~5.6]{grohe2019canonisation} within $O(\log n)$ iterations.

The second such place where Grohe \& Neuen rely on the stable coloring of $1$-WL to detect the connected components of the flipped graph. We will show that $2$-WL can in fact identify these components in $O(\log n)$ rounds. As we are controlling for rounds rather than using the stable coloring, we can further reduce the round complexity via a simple observation. In the pebble game game characterization of WL, if Duplicator fails to respect connected components of the flipped graph, Spoiler can win in $O(\log n)$ rounds \textit{without descending down the rank decomposition}. If Duplicator does respect the connected components, Spoiler only needs a constant number of rounds to descend to a child node in the rank decomposition. In either case, we only need $O(\log n)$ rounds total for Spoiler to win and distinguish the graphs. Thus we reduce the round complexity of WL to $O(\log n)$, which yields a $\textsf{TC}^1$ isomorphism test.

In the process of our work, we came across a result of Bodlaender \cite{Bodlaender}, who showed that any graph of treewidth $k$ admits a \emph{binary} tree decomposition of width $\leq 3k+2$ and height $O(\log n)$. Using Bodlaender's result, we were able to modestly improve the descriptive complexity for graphs of bounded treewidth. 

\begin{theorem} \label{thm:Treewidth}
The $(3k+6)$-dimensional Weisfeiler--Leman algorithm identifies graphs of treewidth $k$ in $O(\log n)$ rounds.
\end{theorem}

In light of the above theorem, we obtain the following improvement in the descriptive complexity for graphs of bounded treewidth.

\begin{corollary} \label{cor:CorTreewidth}
Let $G$ be a graph of treewidth $k$. Then there exists a formula $\varphi_{G} \in \mathcal{C}_{3k+7, O(\log n)}$ that identifies $G$ up to isomorphism. That is, for any $H \not \cong G$, $G \models \varphi_{G}$ and $H \not \models \varphi_{G}$.
\end{corollary}

\section{Preliminaries}

\subsection{Computational Complexity}
We assume familiarity with Turing machines and the complexity classes $\textsf{P}, \textsf{NP}, \textsf{L}$, and $\textsf{NL}$--- we refer the reader to standard references \cite{ComplexityZoo, AroraBarak}. For a standard reference on circuit complexity, see \cite{VollmerText}. We consider Boolean circuits using the \textsf{AND}, \textsf{OR}, \textsf{NOT}, and \textsf{Majority}, where $\textsf{Majority}(x_{1}, \ldots, x_{n}) = 1$ if and only if greater than or equal to $n/2$ of the inputs are $1$. Otherwise, $\textsf{Majority}(x_{1}, \ldots, x_{n}) = 0$. In this paper, we will consider \textit{logspace uniform} circuit families $(C_{n})_{n \in \mathbb{N}}$, in which a deterministic logspace Turing machine can compute the map $1^{n} \mapsto \langle C_{n} \rangle$ (here, $\langle C_{n} \rangle$ denotes an encoding of the circuit $C_{n}$).

\begin{definition}
Fix $k \geq 0$. We say that a language $L$ belongs to (logspace uniform) $\textsf{NC}^{k}$ if there exist a (logspace uniform) family of circuits $(C_{n})_{n \in \mathbb{N}}$ over the $\textsf{AND}, \textsf{OR}, \textsf{NOT}$ gates such that the following hold:
\begin{itemize}
\item The $\textsf{AND}$ and $\textsf{OR}$ gates take exactly $2$ inputs. That is, they have fan-in $2$.
\item $C_{n}$ has depth $O(\log^{k} n)$ and uses (has size) $n^{O(1)}$ gates. Here, the implicit constants in the circuit depth and size depend only on $L$.

\item $x \in L$ if and only if $C_{|x|}(x) = 1$. 
\end{itemize}
\end{definition}

\noindent The complexity class $\textsf{AC}^{k}$ is defined analogously as $\textsf{NC}^{k}$, except that the $\textsf{AND}, \textsf{OR}$ gates are permitted to have unbounded fan-in. That is, a single $\textsf{AND}$ gate can compute an arbitrary conjunction, and a single $\textsf{OR}$ gate can compute an arbitrary disjunction. The complexity class $\textsf{TC}^{k}$ is defined analogously as $\textsf{AC}^{k}$, except that our circuits are now permitted $\textsf{Majority}$ gates of unbounded fan-in. For every $k$, the following containments are well-known:
\[
\textsf{NC}^{k} \subseteq \textsf{AC}^{k} \subseteq \textsf{TC}^{k} \subseteq \textsf{NC}^{k+1}.
\]

\noindent In the case of $k = 0$, we have that:
\[
\textsf{NC}^{0} \subsetneq \textsf{AC}^{0} \subsetneq \textsf{TC}^{0} \subseteq \textsf{NC}^{1} \subseteq \textsf{L} \subseteq \textsf{NL} \subseteq \textsf{AC}^{1}.
\]

\noindent We note that functions that are $\textsf{NC}^{0}$-computable can only depend on a bounded number of input bits. Thus, $\textsf{NC}^{0}$ is unable to compute the $\textsf{AND}$ function. It is a classical result that $\textsf{AC}^{0}$ is unable to compute \algprobm{Parity} \cite{Smolensky87algebraicmethods}. The containment $\textsf{TC}^{0} \subseteq \textsf{NC}^{1}$ (and hence, $\textsf{TC}^{k} \subseteq \textsf{NC}^{k+1}$) follows from the fact that $\textsf{NC}^{1}$ can simulate the \textsf{Majority} gate. The class $\textsf{NC}$ is:
\[
\textsf{NC} := \bigcup_{k \in \mathbb{N}} \textsf{NC}^{k} = \bigcup_{k \in \mathbb{N}} \textsf{AC}^{k} = \bigcup_{k \in \mathbb{N}} \textsf{TC}^{k}.
\]

\noindent It is known that $\textsf{NC} \subseteq \textsf{P}$, and it is believed that this containment is strict.

\subsection{Weisfeiler--Leman}

We begin by recalling the Weisfeiler--Leman algorithm for graphs, which computes an isomorphism-invariant coloring. Let $G$ be a graph on $n$ vertices, let $\chi:V(G)\to [n]$ be a coloring of the vertices, and let $k \geq 2$ be an integer. The $k$-dimensional Weisfeiler--Leman, or $k$-WL, algorithm begins by constructing an initial coloring $\chi_{0} : V(G)^{k} \to \mathcal{K}$, where $\mathcal{K}$ is our set of colors, by assigning each $k$-tuple a color based on its isomorphism type under the coloring $\chi$. Note that for $k$-WL applied to two graphs $G$ and $H$, each of order $n$, there are at most $2n^{k}$ color classes. So without loss of generality, we may take $\mathcal{K} = [2n^{k}]$. Two $k$-tuples $(v_{1}, \ldots, v_{k})\in V(G)^k$ and $(u_{1}, \ldots, u_{k})\in V(G)^k$ receive the same color under $\chi_{0}$ if and only if the following conditions all hold:
\begin{itemize}
    \item For all $i, j$, $v_i = v_j \Leftrightarrow u_i = u_j$. 
    \item The map $v_i \mapsto u_i$ (for all $i \in [k]$) is an isomorphism of the induced subgraphs $G[\{ v_{1}, \ldots, v_{k}\}]$ and $G[\{u_{1}, \ldots, u_{k}\}]$
    \item $\chi(u_i) = \chi(v_i)$ for all $i\in [k]$.
\end{itemize}For $r \geq 0$, the coloring computed at the $r$th iteration of Weisfeiler--Leman is refined as follows. For a $k$-tuple $\overline{v} = (v_{1}, \ldots, v_{k})$ and a vertex $x \in V(G)$, define
\[
\overline{v}(v_{i}/x) = (v_{1}, \ldots, v_{i-1}, x, v_{i+1}, \ldots, v_{k}).
\]

\noindent The coloring computed at the $(r+1)$st iteration, denoted $\chi_{r+1}$, stores the color of the given $k$-tuple $\overline{v}$ at the $r$th iteration, as well as the colors under $\chi_{r}$ of the $k$-tuples obtained by substituting a single vertex in $\overline{v}$ for another vertex $x$. We examine this multiset of colors over all such vertices $x$. This is formalized as follows:
\begin{align*}
\chi_{r+1}(\overline{v}) = &( \chi_{r}(\overline{v}), \{\!\!\{ ( \chi_{r}(\overline{v}(v_{1}/x)), \ldots, \chi_{r}(\overline{v}(v_{k}/x) ) \bigr| x \in V(G) \}\!\!\} ),
\end{align*}
where $\{\!\!\{ \cdot \}\!\!\}$ denotes a multiset. Note that the coloring $\chi_{r}$ computed at iteration $r$ induces a partition of $V(G)^{k}$ into color classes. The Weisfeiler--Leman algorithm terminates when this partition is not refined, that is, when the partition induced by $\chi_{r+1}$ is identical to that induced by $\chi_{r}$. The final coloring is referred to as the \textit{stable coloring}, which we denote $\chi_{\infty} := \chi_{r}$.

The $1$-dimensional Weisfeiler--Leman algorithm, sometimes referred to as \textit{Color Refinement}, works nearly identically. The initial coloring is that provided by the vertex coloring for the input graph. For the refinement step, we have that:
\[
\chi_{r+1}(u) = (\chi_{r}(u), \{\!\!\{ \chi_{r}(v) : v \in N(u) \}\!\!\} ).
\]

\noindent We have that $1$-WL terminates when the partition on the vertices is not refined.

As we are interested in both the Weisfeiler--Leman dimension and the number of rounds, we will use the following notation.

\begin{definition}
Let $k \geq 1$ and $r \geq 1$ be integers. The $(k, r)$-WL  algorithm is obtained by running $k$-WL for $r$ rounds. Here, the initial coloring counts as the first round.
\end{definition}

Let $S$ be a sequence of vertices. The \textit{individualize-and-refine} paradigm works first by assigning each vertex in $S$ a unique color. We then run $(k,r)$-WL starting from this choice of initial coloring. We denote the coloring computed by $(k,r)$-WL after individualizing $S$ as $\chi_{k,r}^{S}$. When there is ambiguity about the graph $G$ in question, we will for clarity write $\chi_{k,r}^{S,G}$.

For two graphs $G$ and $H$, we say that $(k,r)$-WL
\textit{distinguishes} $G$ and $H$ if there is some color $c$ such that the sets
\[
    |\{v\in V(G)^k: \chi_{G,k,r}(v) = c\}| \neq |\{w\in V(H)^k: \chi_{H,k,r}(w) = c\}|.
\] We write $G \equiv_{k,r} H$ if $(k,r)$-WL does not distinguish between G and H. Additionally, $(k,r)$-WL \textit{identifies} a graph $G$ if $(k,r)$-WL distinguishes $G$ from every graph $H$ such that $G\not\cong H$.

\begin{remark}
Grohe \& Verbitsky \cite{GroheVerbitsky} previously showed that for fixed $k$, the classical $k$-dimensional Weisfeiler--Leman algorithm for graphs can be effectively parallelized. Precisely, each iteration (including the initial coloring) can be implemented using a logspace uniform $\textsf{TC}^{0}$ circuit. 
\end{remark}

\subsection{Pebbling Game}

We recall the bijective pebble game introduced by \cite{Hella1989, Hella1993} for WL on graphs. This game is often used to show that two graphs $X$ and $Y$ cannot be distinguished by $k$-WL. The game is an Ehrenfeucht--Fra\"iss\'e game (c.f., \cite{Ebbinghaus:1994, Libkin}), with two players: Spoiler and Duplicator. We begin with $k+1$ pairs of pebbles. Prior to the start of the game, each pebble pair $(p_{i}, p_{i}')$ is initially placed either beside the graphs or on a given pair of vertices $v_{i} \mapsto v_{i}'$ (where $v_{i} \in V(X), v_{i}' \in V(Y)$). We refer to this initial configuration for $X$ as $\overline{v}$, and this initial configuration for $Y$ as $\overline{v'}$.  Each round $r$ proceeds as follows.
\begin{enumerate}
\item Spoiler picks up a pair of pebbles $(p_{i}, p_{i}^{\prime})$. 
\item Duplicator chooses a bijection $f_{r} : V(X) \to V(Y)$ (we emphasize that the bijection chosen depends on the round and, implicitly, the pebbling configuration at the start of said round).
\item Spoiler places $p_{i}$ on some vertex $v \in V(X)$. Then $p_{i}^{\prime}$ is placed on $f(v)$. 
\end{enumerate} 

Let $v_{1}, \ldots, v_{m}$ be the vertices of $X$ pebbled at the end of round $r$ of the game, and let $v_{1}^{\prime}, \ldots, v_{m}^{\prime}$ be the corresponding pebbled vertices of $Y$. Spoiler wins precisely if the map $ v_{\ell} \mapsto v_{\ell}^{\prime}$ is not an isomorphism of the induced subgraphs $X[\{v_{1}, \ldots, v_{m}\}]$ and $Y[\{v_{1}^{\prime}, \ldots, v_{m}^{\prime}\}]$. Duplicator wins otherwise. Spoiler wins, by definition, at round $0$ if $X$ and $Y$ do not have the same number of vertices. We note that $\overline{v}$ and $\overline{v'}$ are not distinguished by the first $r$ rounds of $k$-WL if and only if Duplicator wins the first $r$ rounds of the $(k+1)$-pebble game \cite{Hella1989, Hella1993, CFI}.

We establish a helper lemma, which effectively states that Duplicator must respect connected components of pebbled vertices.

\begin{lemma} \label{lem:Connectivity}
Let $G, H$ be graphs on $n$ vertices. Suppose that $(u, v) \mapsto (u', v')$ have been pebbled. Furthermore, suppose that $u, v$ belong to the same connected component of $G$, while $u', v'$ belong to different connected components of $H$. Then Spoiler can win using $1$ additional pebble and $O(\log n)$ rounds.
\end{lemma}

\begin{proof}
Let $P$ be a shortest $u-v$ path in $G$. Spoiler begins by pebbling a midpoint $w$ of $P$. Let $w'$ be Duplicator's response. As $u', v'$ belong to different components, we may assume without loss of generality that $w'$ belongs to the component containing $u'$. Let $P_{wv}$ be the $w-v$ path within $P$. At the next round, Spoiler picks up the pebble on $u$ and iterates on the above argument using $P_{wv}$. This argument only occurs finitely many times until we hit a base case, where $wv$ is an edge of $G$, while $w' v'$ is not an edge of $H$.

At each round, we are cutting the size of the path in half. Thus, at most $\log_{2}(n) + 1$ rounds are required. Observe that only one additional pebble was used. The result now follows.
\end{proof}

\subsection{Logics} \label{sec:Logics}

We recall key notions of first-order logic. We have a countable set of variables $\{x_{1}, x_{2}, \ldots \}$. Formulas are defined inductively. For the basis, we have that $x_{i} = x_{j}$ is a formula for all pairs of variables. Now if $\varphi_{1}, \varphi_{2}$ are formulas, then so are the following: $\varphi_{1} \land \varphi_{2}, \varphi_{1} \vee \varphi_{2}, \neg{\varphi_{1}}, \exists{x_{i}} \, \varphi_{1},$ and $\forall{x_{i}} \, \varphi_{1}$. In order to define logics on graphs, we add a relation $E(x, y)$, where $E(x,y) = 1$ if and only if $\{x,y\}$ is an edge of our graph, and $0$ otherwise.  In keeping with the conventions of~\cite{CFI}, we refer to the first-order logic with relation $E$ as $\mathcal{L}$ and its $k$-variable fragment as $\mathcal{L}_{k}$. We refer to the logic $\mathcal{C}$ as the logic obtained by adding counting quantifiers $\exists^{\geq n} x \, \varphi$ (there exist at least $n$ elements $x$ that satisfy $\varphi$) and $\exists{!n} \, x \, \varphi$ (there exist exactly $n$ elements $x$ that satisfy $\varphi$) and its $k$-variable fragment as $\mathcal{C}_{k}$. 

The \textit{quantifier depth} of a formula $\varphi$ (belonging to either $\mathcal{L}$ or $\mathcal{C}$) is the depth of its quantifier nesting. We denote the quantifier depth of $\varphi$ as $\text{qd}(\varphi)$. This is defined inductively as follows.
\begin{itemize}
    \item If $\varphi$ is atomic, then $\text{qd}(\varphi) = 0$.

    \item $\text{qd}(\neg \varphi) = \text{qd}(\varphi)$.

    \item $\text{qd}(\varphi_{1} \vee \varphi_{2}) = \text{qd}(\varphi_{1} \land \varphi_{2}) = \max\{ \text{qd}(\varphi_{1}), \text{qd}(\varphi_{2})\}$.

    \item $\text{qd}(Qx \, \varphi) = \text{qd}(\varphi) + 1$, where $Q$ is a quantifier in the logic.
\end{itemize}

\noindent We denote the fragment of $\mathcal{L}_{k}$ (respectively, $\mathcal{C}_{k}$) where the formulas have quantifier depth at most $r$ as $\mathcal{L}_{k,r}$ (respectively, $\mathcal{C}_{k,r}$). Let $\overline{v} \in V(X)^{k}, \overline{v'} \in V(Y)^{k}$. We note that $\overline{v}, \overline{v'}$ are distinguished by $(k,r)$-WL if and only if there exists a formula $\varphi \in \mathcal{C}_{k+1,r}$ such that $(X, \overline{v}) \models \varphi$ and $(Y, \overline{v'}) \not \models \varphi$ \cite{ImmermanLander1990, CFI}.

\subsection{Rank-Width}

\noindent Oum \& Seymour \cite{OumSeymour2006} introduced the rank-width parameter to measure the width of a certain hierarchical decomposition of graphs. The goal is to intuitively split the vertices of a graph along cuts of low complexity in a hierarchical fashion. Here, the complexity is the $\mathbb{F}_{2}$-rank of the matrix capturing the adjacencies crossing the cut.

Precisely, let $G$ be a graph, and let $X, Y \subseteq V(G)$. Define $M(X, Y) \in \mathbb{F}_{2}^{X \times Y}$ to be the matrix where $(M(X,Y))_{uv} = 1$ if and only if $uv \in E(G)$. That is, $M(X,Y)$ is the submatrix of the adjacency matrix whose rows are indexed by $X$ and whose columns are indexed by $Y$. Denote $\rho(X) := \rk_{\mathbb{F}_{2}}(M(X, \overline{X}))$.

A \textit{rank decomposition} of $G$ is a tuple $(T, \gamma)$, where $T$ is a rooted binary tree and $\gamma : V(T) \to 2^{V(G)}$ satisfying the following:
\begin{itemize}
    \item For the root $r$ of $T$, $\gamma(r) = V(G)$.
    \item For an internal node $t \in V(T)$, denote the children of $t$ as $s_{1}, s_{2}$. For every internal node $t$, we have that $\gamma(t) = \gamma(s_{1}) \cup \gamma(s_{2})$, and $\gamma(s_{1}) \cap \gamma(s_{2}) = \emptyset$.
    
    \item For any leaf $t \in V(T)$, $|\gamma(t)| = 1$.
\end{itemize}

\begin{remark}
Let $L(T)$ be the set of leaves of $T$. Instead of providing $\gamma$, we can equivalently define a bijection $f : V(G) \to L(T)$. By the second condition of a rank decomposition, $f$ completely determines $\gamma$.
\end{remark}

\noindent The \textit{width} of a rank decomposition $(T, \gamma)$ is:
\[
\text{wd}(T, \gamma) := \max\{ \rho_{G}(\gamma(t)) : t \in V(T) \}.
\]

\noindent The \textit{rank-width} of a graph $G$ is:
\[
\rw(G) := \min\{ \text{wd}(T, \gamma) : (T, \gamma) \text{ is a rank decomposition of } G \}.
\]

\noindent The parameter rank-width is closely related to the parameter clique width, introduced by Courcelle \& Olariu \cite{CourcelleOlariu}. Oum \& Seymour  \cite{OumSeymour2006} showed that:
\[
\rw(G) \leq \text{cw}(G) \leq 2^{\rw(G) + 1} - 1.
\]

\noindent Denote $\text{tw}(G)$ to be the treewidth of $G$. Oum \cite{Oum2008} showed that $\rw(G) \leq \text{tw}(G) + 1$. Note that $\text{tw}(G)$ cannot be bounded in terms of $\rw(G)$; for instance, the complete graph $K_{n}$ has $\rw(K_{n}) = 1$ but $\text{tw}(K_{n}) = n-1$.

\section{Weisfeiler--Leman for Graphs of Bounded Rank-Width} \label{sec:WLRankWidth}

\subsection{Split Pairs and Flip Functions}

In designing a pebbling strategy for graphs of bounded rank-width, Grohe \& Neuen \cite{grohe2019canonisation} sought to pebble a set of vertices $X \subseteq V(G)$ such that $\rho(X) \leq k$ and pebbling $X$ partitions the remaining vertices into sets $C_{1}, \ldots, C_{\ell}$ that can be treated independently. Furthermore, we want for each $i \in [\ell]$ that either $C_{i} \subseteq X$ or $C_{i} \subseteq \overline{X}$. As there can be many edges between $X$ and $\overline{X}$, this is hard to accomplish in general. To this end, Grohe \& Neuen \cite{grohe2019canonisation} utilized split pairs and flip functions. We will now recall their framework.

Let $G(V,E, \chi)$ be a colored graph on $n$ vertices, and suppose the rank-width of $G$ is at most $k$. Let $X \subseteq V(G)$. For $v \in X$, define $\text{vec}_{X}(v) = (a_{v,w})_{w \in \overline{X}} \in \mathbb{F}_{2}^{\overline{X}}$, where $a_{v,w} = 1$ if and only if $vw \in E(G)$. For $S \subseteq X$, define $\text{vec}_{X}(S) = \{ \text{vec}_{X}(v) : v \in S \}$. A \textit{split pair} for $X$ is a pair $(A, B)$ such that:
\begin{enumerate}[label=(\alph*)]
\item $A \subseteq X$, and $B \subseteq \overline{X}$,

\item $\text{vec}_{X}(A)$ forms a linear basis for $\langle \text{vec}_{X}(X) \rangle$, and
\item $\text{vec}_{\overline{X}}(B)$ forms a linear basis for $\langle \text{vec}_{\overline{X}}(\overline{X}) \rangle$.
\end{enumerate}

An \textit{ordered split pair} for $X$ is a pair $((a_1, \ldots, a_q), (b_1, \ldots, b_p))$ such that $(\{a_1, \ldots, a_q\}, \{b_1, \ldots, b_p\})$ is a split pair for $X$.

For $v,w \in V(G)$, we say that $v \approx_{(\overline{a}, \overline{b})} w$ if $N(v) \cap (\overline{a}, \overline{b}) = N(w) \cap (\overline{a}, \overline{b})$ (here, we consider $N(v) \cap (\overline{a}, \overline{b})$ as a set). Observe that $\approx_{(\overline{a}, \overline{b})}$ forms an equivalence relation. For $(\overline{a}, \overline{b}) \in V(G)^{\leq 2k}$, let $2^{\overline{a} \cup \overline{b}}$ be the set of all subsets of $\overline{a} \cup \overline{b} \subseteq V(G)$, where we abuse notation by considering $\overline{a}, \overline{b}$ as subsets of $V(G)$. A \textit{flip extension} of an ordered split pair $(\overline{a}, \overline{b})$ is a tuple:
\[
\overline{s} := \biggr( \overline{a}, \overline{b}, f : \left( 2^{\overline{a} \cup \overline{b}} \right)^2 \to [n] \cup \{ \perp \} \biggr),
\]

\noindent such that for all $M, N \in 2^{\overline{a} \cup \overline{b}}$ with $M \neq N$, either $f(M,N) = \perp$ or $f(N,M) = \perp$. There is no restriction on $f(M,N)$ if $M = N$. For $v, w \in V(G)$, we say that $v \approx_{\overline{s}} w$ if $v \approx_{(\overline{a}, \overline{b})} w$. Denote $[v]_{\approx \overline{s}}$ to be the equivalence class of $v$ with respect to $\approx_{\overline{s}}$. Define the \textit{flipped graph} $G^{\overline{s}} = (V, E^{\overline{s}}, \chi, \overline{a}, \overline{b})$, where $V(G^{\overline{s}}) = V(G)$,
\begin{align*}
E^{\overline{s}} &:= \{ vw \in E(G) : f( N(v) \cap (\overline{a}, \overline{b}), N(w) \cap (\overline{a}, \overline{b})) = d \in [n] \land |N(v) \cap [w]_{\approx \overline{s}}| <  d \} \\
&\cup	\{ vw \not \in E(G) : f( N(v) \cap (\overline{a}, \overline{b}), N(w) \cap (\overline{a}, \overline{b})) = d \in [n] \land |N(v) \cap [w]_{\approx \overline{s}}| \geq d \},
\end{align*}

\noindent and $\chi$ is the same coloring as in $G$. Let $\text{Comp}(G, \overline{s}) \subseteq 2^{V(G)}$ be the set of vertex sets of the connected components of $G^{\overline{s}}$. Observe that $\text{Comp}(G, \overline{s})$ forms a partition of $V(G)$.

Grohe \& Neuen \cite{grohe2019canonisation} established that for any choice $(\overline{a}, \overline{b})$ of split pair, there exists a suitable flip function; and thus, a suitable flip extension. 

\begin{lemma}[{\cite[Lemma~5.6]{grohe2019canonisation}}] \label{lem:GNLem5.6}
Let $G$ be a (colored) graph, and let $X \subseteq V(G)$. Furthermore, let $(\overline{a}, \overline{b})$ be an ordered split pair for $X$. Then there exists a flip extension $\overline{s} := (\overline{a}, \overline{b}, f)$ such that $C \subseteq X$ or $C \subseteq \overline{X}$ for every $C \in \text{Comp}(G, \overline{s})$.
\end{lemma}

Grohe \& Neuen \cite[Section~5]{grohe2019canonisation} considered uncolored flipped graphs. As the conditions for determining the edges of the flipped graph do not depend on the vertex colors, \cite[Lemma~5.6]{grohe2019canonisation} holds in our setting. 

We now turn to showing that the flip extensions preserve both isomorphism and the effects of Weisfeiler--Leman. To do so, we consider vertex colorings $\chi$ that refine the coloring $\chi_{1,3}$ computed by $(1,3)$-WL. The advantage of incorporating such a coloring on the vertices is that it encodes some data about how the vertices of $G$ interact with the specified split pair. Furthermore, the colorings computed by Weisfeiler--Leman are invariant under isomorphism. We take advantage of this to establish that the flipped graph preserves both the isomorphism problem (Lemma~\ref{lem:PreservesIsomorphism}) and the effects of Weisfeiler--Leman (Lemma~\ref{lem:PreservesWL}). For a  graph $G$ of rank-width $k$, we will be running $(6k+3, O(\log n))$, and so we may assume without loss of generality that the vertices of $G$ have been colored according to $(1,3)$-WL.

\begin{lemma} \label{lem:PreservesIsomorphism}
Let $G, H$ be graphs, and let $\overline{s} = (\overline{a}, \overline{b}, f), \overline{s'} = (\overline{a'}, \overline{b'}, f)$ be flip extensions for $G, H$, respectively (we stress that the function $f$ appearing in $\overline{s}$ is the same as that appearing in $\overline{s'}$). Let $k \geq 1, r \geq 3$. Consider the colorings $\chi_{k,r}^{(\overline{a}, \overline{b}), G}, \chi_{k,r}^{(\overline{a'}, \overline{b'}), H}$ obtained by individualizing $(\overline{a}, \overline{b}) \mapsto (\overline{a'}, \overline{b'})$ and applying $(k,r)$-WL. 

Let $\varphi : V(G) \to V(H)$ be a bijection. We have that $\varphi$ is an isomorphism of the colored graphs $(G, \chi_{k,r}^{(\overline{a}, \overline{b}), G}) \cong (H, \chi_{k,r}^{(\overline{a'}, \overline{b'}), H})$ if and only if $\varphi$ is an isomorphism of  $G^{\overline{s}} \cong H^{\overline{s'}}$. 
\end{lemma}

\begin{proof}
Suppose first that $(G, \chi_{k,r}^{(\overline{a}, \overline{b}), G}) \cong (H, \chi_{k,r}^{(\overline{a'}, \overline{b'}), H})$, and let $\varphi : V(G) \to V(H)$ be an isomorphism of $(G, \chi_{k,r}^{(\overline{a}, \overline{b}, G})$ and $(H, \chi_{k,r})$. We claim that $\varphi$ is also an isomorphism of $G^{\overline{s}}$ and $H^{\overline{s'}}$. We first note that for each vertex $v$ in $G^{\overline{s}}$ (resp. $H^{\overline{s'}}$), $v$ receives the same color in both $G^{\overline{s}}$ and $G$ (resp. $H^{\overline{s'}}$ and $H$). Thus, as $\varphi$ is a colored graph isomorphism of $(G, \chi_{k,r}^{(\overline{a}, \overline{b}), G})$ and $(H, \chi_{k,r}^{(\overline{a'}, \overline{b'}), H})$, we have that $\varphi$ respects the vertex colors of $G^{\overline{s}}$ and $H^{\overline{s'}}$. In particular, as $\chi_{k,r}^{(\overline{a}, \overline{b}), G}, \chi_{k,r}^{(\overline{a'}, \overline{b'}), H}$ are isomorphism invariant, we have that for all $v \in V(G)$: $\chi_{k,r}^{(\overline{a'}, \overline{b'}), H}(\varphi(v)) = \chi_{k,r}^{(\overline{a}, \overline{b}), G}(v)$.

It remains to show that for all $v,w \in V(G) (= V(G^{\overline{s}}))$, $vw \in E(G^{\overline{s}})$ if and only if $\varphi(v)\varphi(w) \in E(H^{\overline{s'}})$. Fix $v, w \in V(G)$. As $\varphi : (G, \chi_{k,r}^{(\overline{a}, \overline{b}), G}) \cong (H, \chi_{k,r}^{(\overline{a'}, \overline{b'}), H})$, we have that:
\begin{align*}
&f( N(v) \cap (\overline{a}, \overline{b}), N(w) \cap (\overline{a}, \overline{b})) = f( N(\varphi(v)) \cap (\overline{a'}, \overline{b'}), N(\varphi(w)) \cap (\overline{a'}, \overline{b'})), \text{ and} \\
&|N(v) \cap [w]_{\approx \overline{s}}| = |N(\varphi(v)) \cap [\varphi(w)]_{\approx \overline{s}}|.
\end{align*}

\noindent It follows immediately that $vw \in E(G^{\overline{s}})$ if and only if $\varphi(v)\varphi(w) \in E(H^{\overline{s'}})$. So $\varphi : G^{\overline{s}} \cong H^{\overline{s}}$.  

Conversely, suppose that $\varphi : G^{\overline{s}} \cong H^{\overline{s'}}$. Suppose to the contrary that $\varphi$ is not an isomorphism of $G$ and $H$. By similar argument as above, we may assume that for all $v \in V(G)$: $\chi_{k,r}^{(\overline{a'}, \overline{b'}), H}(\varphi(v)) = \chi_{k,r}^{(\overline{a}, \overline{b}), G}(v)$. Thus, there must exist $v, w \in V(G)$ such that, without loss of generality, $vw \in E(G)$, but $\varphi(v)\varphi(w) \not \in E(H)$. As $\varphi : G^{\overline{s}} \cong H^{\overline{s'}}$, we have that $vw \in E(G^{\overline{s}})$ if and only if $\varphi(v)\varphi(w) \in E(H^{\overline{s}})$. Without loss of generality, suppose that $vw \in E(G^{\overline{s}})$ (the argument is essentially identical in the case when $vw \not \in E(G^{\overline{s}})$). As $vw \in E(G)$ and $vw \in E(G^{\overline{s}})$, we have that there exists some $d \in [n]$ such that: 
\begin{align*}
&f( N(v) \cap (\overline{a}, \overline{b}), N(w) \cap (\overline{a}, \overline{b})) = d \text{ and}, \\ 
&|N(v) \cap [w]_{\approx \overline{s}}| <  d.
\end{align*}

\noindent However, as $\varphi(v)\varphi(w) \in E(H^{\overline{s}})$ and $\varphi(v)\varphi(w) \not \in E(H)$, there exists some $d' \in [n]$ such that:
\begin{align*}
&f( N(\varphi(v)) \cap (\overline{a'}, \overline{b'}), N(\varphi(w)) \cap (\overline{a'}, \overline{b'})) = d' \text{ and},\\ 
&|N(\varphi(v)) \cap [\varphi(w)]_{\approx \overline{s'}}| \geq d'.
\end{align*}

\noindent Observe that if for $u \in \{v,w\}$, $\varphi(N(u) \cap (\overline{a}, \overline{b})) \neq N(\varphi(u)) \cap (\overline{a'}, \overline{b'})$, then after individualizing $(\overline{a}, \overline{b}) \mapsto (\overline{a'}, \overline{b'})$ and running $(k,2)$-WL, that $u$ and $\varphi(u)$ will receive different colors. However, $\chi_{k,r}^{(\overline{a}, \overline{b}), G}(u) = \chi_{k,r}^{(\overline{a'}, \overline{b'}), H}(\varphi(u))$. Thus, we may assume that $d = d'$. It follows that (without loss of generality):
\begin{align*}
&|N(v) \cap [w]_{\approx \overline{s}}| < d, \text{ while} \\
&|N(\varphi(v)) \cap [\varphi(w)]_{\approx \overline{s}'}| \geq d.
\end{align*}

\noindent It follows immediately from the definition of $\approx_{s'}, \approx_{\overline{s'}}$ that the following conditions hold (for any $r \geq 2$):
\begin{itemize}
\item if $x \in [w]_{\approx \overline{s}}$ and $y \in V(G) \setminus [w]_{\approx \overline{s}}$, then $\chi_{k,r}(x) \neq \chi_{k,r}(y)$; and
\item if $x \in [w]_{\approx \overline{s}}$ and $y \in V(H) \setminus [\varphi(w)]_{\approx \overline{s'}}$, then $\chi_{k,r}(x) \neq \chi_{k,r}(y)$.
\end{itemize}

\noindent But then at iteration $r = 3$, $k$-WL will assign $v$ and $\varphi(v)$ different colors, a contradiction. Thus, $\varphi$ must be an isomorphism of the colored graphs $(G, \chi_{k,r}^{(\overline{a}, \overline{b}, G})$ and $(H, \chi_{k,r}^{(\overline{a'}, \overline{b'}, H})$.
\end{proof}

\begin{lemma}[cf. {\cite[Lemma~3.10]{grohe2019canonisation}}] \label{lem:PreservesWL}
Let $G(V,E,\chi), G'(V', E', \chi')$ be colored graphs, and let $\overline{s} = (\overline{a}, \overline{b}, f)$ and $\overline{s'} = (\overline{a'}, \overline{b'}, f)$ be flip extensions (we are using the same flip function $f$ for both $\overline{s}, \overline{s'}$). Let $\chi_{1,3}$ be the coloring resulting from individualizing $(\overline{a}, \overline{b}) \mapsto (\overline{a'}, \overline{b'})$ and running $(1,3)$-WL. Suppose that $\chi, \chi'$ both refine $\chi_{1,3}$.

Let $((\overline{v}, \overline{w})) = ((v_1, \ldots, v_{\ell}), (w_{1}, \ldots, w_{\ell}))$ be a position in the $\ell$-pebble bijective pebble game. We have that Spoiler wins from $((\overline{v}, \overline{w}))$ in the $\ell$-pebble, $r$-round game on $(G, G')$ if and only if Spoiler wins from from $((\overline{v}, \overline{w}))$ in the $\ell$-pebble, $r$-round game on $(G^{\overline{s}}, (G')^{\overline{s}})$.
\end{lemma}

\begin{proof}
Let $\overline{v}, \overline{w}$ be configurations in the $\ell$-pebble game. We have that the map $v_{i} \mapsto w_{i}$ is a marked isomorphism of the induced subgraphs $G[\overline{v}]$ and $G'[\overline{w}]$ if and only if this map is an isomorphism of the induced subgraphs $G^{\overline{s}}[\overline{v}]$ and $(G')^{\overline{s}}[\overline{w}]$.
\end{proof}

\begin{corollary}[Compare rounds cf. {\cite[Corollary~3.12]{grohe2019canonisation}}] \label{cor:GN3.12}
Let $G(V,E,\chi), G'(V', E', \chi')$ be colored graphs, and let $\overline{s} = (\overline{a}, \overline{b}, f)$ and $\overline{s'} = (\overline{a'}, \overline{b'}, f)$ be flip extensions (we are using the same flip function $f$ for both $\overline{s}, \overline{s'}$). Let $\chi_{1,3}$ be the coloring resulting from individualizing $(\overline{a}, \overline{b}) \mapsto (\overline{a'}, \overline{b'})$ and running $(1,3)$-WL. Suppose that $\chi, \chi'$ both refine $\chi_{1,3}$.

Let $\overline{v} \in V^{k}, \overline{v'} \in (V')^{k}$. Let $C$ be a connected component of $G^{\overline{s}}$ such that $\chi(u) \neq \chi(w)$ for all $u \in C$ and all $w \in V \setminus C$. Let $C'$ be a connected component of $(G')^{\overline{s'}}$ such that $\chi'(u') \neq \chi'(w')$ for all $u' \in C'$ and $w' \in V' \setminus C'$. Let $r \geq 1$. Suppose that:
\[
(G[C], \chi_{1,r}^{\overline{v}, G}) \not \cong (G'[C'], \chi_{1,r}^{\overline{v'}, G'}).
\]

\noindent Let $\overline{w} := C \cap \overline{v}$ and $\overline{w'} := C' \cap \overline{v'}$. Then either:
\[
(G[C], \chi_{1,r}^{\overline{w}, G}) \not \cong (G'[C'], \chi_{1,r}^{\overline{w'}, G'}), 
\]

\noindent or $r$ rounds of Color Refinement distinguishes $(G, \chi^{\overline{v}})$ from $(G', (\chi')^{\overline{v'}})$.
\end{corollary}

\begin{proof}
The proof is by contrapositive. Let $I = \{ i \in [k] : v_{i} \in C\}$, and let $I' = \{ i \in [k] : v_{i}' \in C'\}$. Suppose that $(1,r)$-WL fails to distinguish $(G, \chi^{\overline{v}})$ and $(G', (\chi')^{\overline{v'}})$. Then by \Lem{lem:PreservesIsomorphism} and \Lem{lem:PreservesWL}, $(1,r)$-WL fails to distinguish $(G^{\overline{s}}, \chi^{\overline{v}})$ and $((G')^{\overline{s}}, (\chi')^{\overline{v'}})$. So we have that $I = I'$. Now suppose that:
\[
(G[C], \chi_{1,r}^{\overline{w}, G}) \cong (G[C'], \chi_{1,r}^{\overline{w'}, G'}). 
\]

\noindent As $I = I'$, it follows that:
\[
(G[C], \chi_{1,r}^{\overline{v}, G}) \cong (G[C'], \chi_{1,r}^{\overline{v'}, G'}).
\]

\noindent As $1$-WL only takes into account the neighbors of a given vertex in the refinement step, we have by induction that for each $i \leq r$:
\[
(G[C], \chi_{1,i}^{\overline{v}, G}) \cong (G[C'], \chi_{1,i}^{\overline{v'}, G'}).
\]

\noindent The result now follows.
\end{proof}

\subsection{WL for Graphs of Bounded Rank-Width}

Our goal in this section is to establish the following.

\begin{theorem} \label{thm:WLRankWidth}
Let $G$ be a graph on $n$ vertices of rank-width $k$, and let $H$ be an arbitrary graph such that $G \not \cong H$. We have that the $(6k+3, O(\log n))$-WL algorithm will distinguish $G$ from $H$.
\end{theorem}

\begin{definition}[{\cite[Definition~4.1]{grohe2019canonisation}}]
Let $G$ be a graph, and let $X, X_{1}, X_{2} \subseteq V(G)$ such that $X = X_{1} \sqcup X_{2}$. Let $(A, B)$ be a split pair for $X$, and let $(A_{i}, B_{i})$ ($i = 1, 2)$ be a split pair for $X_{i}$. We say that $(A_{i}, B_{i})$ are \textit{nice} with respect to $(A, B)$ if the following conditions hold:
\begin{enumerate}[label=(\alph*)]
\item $A \cap X_{i} \subseteq A_{i}$ for each $i \in \{1,2\}$, and 
\item $B_{2} \cap X_{1} \subseteq A_{1}$ and similarly $B_{1} \cap X_{2} \subseteq A_{2}$. 
\end{enumerate}

\noindent A triple $((A, B), (A_{1}, B_{1}), (A_{2}, B_{2}))$ of ordered split pairs is \textit{nice} if the underlying triple of unordered split pairs is nice.
\end{definition}

\begin{lemma}[{\cite[Lemma~4.2]{grohe2019canonisation}}] \label{lem:ExistNice}
Let $G$ be a graph, and let $X, X_{1}, X_{2} \subseteq V(G)$ such that $X = X_{1} \sqcup X_{2}$. Let $(A, B)$ be a split pair for $X$. There exist nice split pairs $(A_{i}, B_{i})$ for $X_{i}$ ($i = 1, 2$) such that additionally $B_{i} \cap \overline{X_{i}} \subseteq B$.
\end{lemma}

\begin{remark}
Grohe \& Neuen use in the proof of \cite[Theorem~5.5]{grohe2019canonisation} that \cite[Lemma~4.2]{grohe2019canonisation} holds for the flipped graphs they define in Section 5, and not just the earlier notion of flipped graphs they consider in Section 3. Hence, \cite[Lemma~4.2]{grohe2019canonisation} holds in our setting as well.
\end{remark}

\begin{definition}
Let $G$ be a graph. A \textit{component partition} of $G$ is a partition $\mathcal{P}$ of $V(G)$ such that every connected component appears in exactly one block of $\mathcal{P}$. That is, for every connected component $C$ of $G$, there exists a $P \in \mathcal{P}$ such that $C \subseteq P$.
\end{definition}

\begin{lemma}[{\cite[Observation~4.3]{grohe2019canonisation}}] \label{Observation4.3}
Let $G, H$ be two non-isomorphic graphs, and let $\mathcal{P}, \mathcal{Q}$ be component partitions of $G, H$, respectively. Let $\sigma : V(G) \to V(H)$ be a bijection. There exists a vertex $v$ of $G$ such that $G[P] \not \cong H[Q]$, where $P \in \mathcal{P}$ is the unique set containing $v$ and $Q \in \mathcal{Q}$ is the unique set containing $\sigma(v)$.
\end{lemma}

We now prove \Thm{thm:WLRankWidth}.

\begin{proof}[Proof of \Thm{thm:WLRankWidth}]
We follow the strategy of \cite[Theorem~4.4]{grohe2019canonisation}. Let $G(V, E, \chi_{G})$ be a colored graph of rank-width $\leq k$, and let $H$ be an arbitrary graph such that $G \not \cong H$. By \cite[Theorem~5]{CourcelleKante2007}, $G$ admits a rank decomposition $(T, \gamma)$ of width at most $2k$ where $T$ has height at most $3 \cdot (\log(n) + 1)$.

We will show that Spoiler has a winning strategy in the pebble game with $6k+3$ pebbles on the board, and in $O(\log n)$ rounds. In a similar manner as in the proof of \cite[Theorem~4.4]{grohe2019canonisation}, we will first argue that $12k+5$ pebbles suffice, and then show how to improve the bound to use only $6k+3$ pebbles.

Spoiler's strategy is to play along the rank decomposition $(T, \gamma)$ starting from the root. As Spoiler proceeds down the tree, the non-isomorphism is confined to increasingly smaller parts of $G$ and $H$. At a node $t \in V(T)$, Spoiler pebbles a split pair $(\overline{a}, \overline{b})$ of $X = \gamma(t)$. Let $\overline{s} = (\overline{a}, \overline{b}, \varphi), \overline{s'} = (\overline{a'}, \overline{b'}, \varphi)$ be the flip extensions provided by Lemma~\ref{lem:GNLem5.6} (we stress that the same flip function $\varphi$ is used for both $\overline{s}, \overline{s'}$). We now turn to confining the non-isomorphism. After the first three rounds of Color Refinement, Spoiler identifies a pair of non-isomorphic components $C \subseteq X, C' \subseteq V(H)$ in the flipped graphs $G^{\overline{s}}$ and $H^{\overline{s'}}$. In particular, Spoiler seeks to find such components $C$ and $C'$ such that $C$ is increasingly further from the root of $T$. Once Spoiler reaches a leaf node of $T$, Spoiler can quickly win. Spoiler places a pebble on a vertex in $C$ and its image in $C'$, under Duplicator's bijection at the given round.

We note that three rounds of Color Refinement suffice for WL to detect the partitioning induced by the flip function (see the proof of Lemma~\ref{lem:PreservesIsomorphism}), though it is not sufficiently powerful to detect the connected components of $G^{\overline{s}}$ and $H^{\overline{s}}$. In the argument below, we will technically consider graphs where the refinement step uses $(2, O(\log n))$-WL. This ensures that, after individualizing a vertex on a given component $C$, the vertices of $C$ receive different colors than those of $V(G) \setminus C$. This will eventually happen, and so in the pebble game characterization, we can continue to descend along $T$ as if the vertices of $C$ have been distinguished from $V(G) \setminus C$. This is a key point where our strategy deviates from that of \cite[Theorem~4.4]{grohe2019canonisation}.

Suppose that at a given round, we have the pebbled configuration $((\overline{a}, \overline{b}, v), (\overline{a'}, \overline{b'}, v'))$, where:
\begin{itemize}
    \item There exists a $t \in T$ such that  $(\overline{a}, \overline{b})$ is a split pair for $\gamma(t)$,
    \item $v \in \gamma(t)$, and let $C$ be the component containing $v$; and 
    \item $v'$ belongs to some component $C'$ of $H$ such that:
    \[
    (G[C], \chi_{2,O(\log n)}^{(\overline{a}, \overline{b}, v)}) \not \cong (H[C'], \chi_{2,O(\log n)}^{(\overline{a'}, \overline{b'}, v)}).   
    \]
\end{itemize}

\noindent \\ Observe that this configuration uses at most $4k+1$ pebbles ($2k$ pebbles for $\overline{a}$, $2k$ pebbles for $\overline{b}$, and a single pebble for $v$). Now by \Lem{lem:GNLem5.6}, we may assume that $C \subseteq \gamma(t)$. Grohe \& Neuen \cite[Theorem~4.4]{grohe2019canonisation} observed that $\chi_{\infty}^{(\overline{a}, \overline{b}, v)}(u_{1}) \neq \chi_{\infty}^{(\overline{a}, \overline{b}, v)}(u_{2})$ for all $u_{1} \in C$ and all $u_{2} \in V(G) \setminus C$. However, in order to obtain the analogous result using $O(\log n)$ rounds, Lemma~\ref{lem:Connectivity} provides that Weisfeiler--Leman of dimension $k \geq 2$ suffices. Note that we will be using WL of dimension $\geq 9$. Furthermore, as $T$ has height $3 \cdot (\log(n) + 1)$, we will be running $(6k+3)$-WL for $\geq 3 \cdot (\log(n)+1)$ rounds. Thus, we may assume without loss of generality that Duplicator selects bijections that map $C \mapsto C'$ (setwise). 

By definition, if $t \in V(T)$ is the root node, then $\gamma(t) = V(G)$. So the empty configuration is a split pair for $\gamma(t)$. We now show by induction on $|\gamma(t)|$ that Spoiler can win from such a position.

Suppose $t$ is a leaf node. Then $|\gamma(t)| = 1$. In this case, $C = \{ v\}$. If $(G[C], \chi_{1,3}^{(\overline{a}, \overline{b}, v)}) \not \cong (H[C'], \chi_{1,3}^{(\overline{a'}, \overline{b'}, v')})$, then either (i) $v$ and $v'$ are assigned different colors, or (ii) $|C'| > 1$. In either case, Spoiler wins with at most $1$ additional pebble and $2$ additional rounds.

For the inductive step, suppose $|\gamma(t)| > 1$. Let $t_{1}, t_{2}$ be the children of $t$ in $T$, and let $X_{i} := \gamma(t_{i})$ ($i = 1, 2$). By \Lem{lem:ExistNice}, there exist nice split pairs $(\overline{a}_{i}, \overline{b}_{i})$ for $t_{i}$. Spoiler pebbles $(\overline{a}_{1}, \overline{b}_{1}, \overline{a}_{2}, \overline{b}_{2})$. Let $(\overline{a'_{1}}, \overline{b'_{1}}, \overline{a'_{2}}, \overline{b'_{2}})$ be Duplicator's response. Let $\alpha := (\overline{a}, \overline{b}, \overline{a}_{1}, \overline{b}_{1}, \overline{a}_{2}, \overline{b}_{2}, v)$, and let $\alpha' := (\overline{a'}, \overline{b'}, \overline{a'_{1}}, \overline{b'_{1}}, \overline{a'_{2}}, \overline{b'_{2}}, v')$. As Grohe \& Neuen \cite[Theorem~4.1]{grohe2019canonisation} noted, the intuitive advantage of pebbling nice split pairs is that we can remove the pebbles from $\overline{a}, \overline{b}$ and $\overline{a}_{3-i}, \overline{b}_{3-i}$ without unpebbling some element of $X_{i}$. We will use this later in our analysis.

Let $\overline{s_{i}} = (\overline{a}_{i}, \overline{b}_{i}, \varphi_{i})$ be the flip extension with respect to the split pair $(\overline{a}_{i}, \overline{b}_{i})$ for $X_{i}$ obtained via Lemma~\ref{lem:GNLem5.6}. Define $\overline{s_{i}'} = (\overline{a'}_{i}, \overline{b'}_{i}, \varphi_{i})$ analogously (here, we stress that $\varphi_{i}$ is the same flip function in both $\overline{s_{i}}, \overline{s_{i}'}$). Let $f : V(G) \to V(H)$ be the bijection that Duplicator selects. As $v \mapsto v'$ has been pebbled, $v \in C$, and $v' \in C'$, we may assume that $f(C) = C'$. Otherwise, there exists a vertex $w \in C$ such that $f(w) \not \in C'$. Spoiler places a pebble on $w \mapsto f(w)$. Now by \Lem{lem:Connectivity}, Spoiler wins using one more pebble and $O(\log n)$ rounds.

Now without loss of generality, suppose $v \in X_{1}$. Let $S := \{ C_{1}, \ldots, C_{p} \}$ be the components of $\text{Comp}(G, \overline{s}_{1})$ which have a non-empty intersection with $C$. Let $S' := \{ C_{1}', \ldots, C_{p}'\}$ be the analogous set of components in $\text{Comp}(H, \overline{s'}_{1})$ that have non-empty intersection with $C'$. By \Lem{lem:PreservesIsomorphism}, we have that:
\[
( (G^{\overline{s_1}})[C], \chi_{2,O(\log n)}^{\overline{\alpha}}) \not \cong ( (H^{\overline{s_{1}'}})[C'], \chi_{2,O(\log n)}^{\overline{\alpha'}}).
\]

\noindent Let $f' : V(G) \to V(H)$ be the bijection that Duplicator selects. Now by \Lem{Observation4.3}, there exists some $w \in C$ such that 
\begin{align} \label{eq:BeforeCase1}
( (G^{\overline{s_{1}}})[C \cap C_{i}], \chi_{2,O(\log n)}^{\overline{\alpha}}) \not \cong ( (H^{\overline{s_{1}'}})[C' \cap C_{i}'], \chi_{2,(O(\log n)}^{\overline{\alpha'}}),
\end{align}

\noindent where $i \in [p]$ is the unique index such that $C_{i}$ contains $w$. We label $C_{i}'$ to be the unique component of $S'$ containing $f'(w)$. By \Lem{lem:PreservesIsomorphism}, we have that:
\[
(G[C \cap C_{i}], \chi_{2,O(\log n)}^{\overline{\alpha}}) \not \cong (H[C' \cap C_{i}'], \chi_{2,O(\log n)}^{\overline{\alpha'}}).
\]

\noindent By \Lem{lem:GNLem5.6}, we have that $C_{i} \subseteq X_{1}$ or $C_{i} \subseteq \overline{X_{1}}$. In particular, $C \cap C_{i} \subseteq X_{1}$ or $C \cap C_{i} \subseteq X_{2}$. We consider the following cases.
\begin{itemize}
\item \textbf{Case 1:} Suppose that $C \cap C_{i} \subseteq X_{1}$. As the configuration $\overline{\alpha}$ has been pebbled, we have by Lemma~\ref{lem:Connectivity} that $\chi_{2,O(\log n)}^{\overline{\alpha}}((u_{1}, u_{1})) \neq \chi_{2,O(\log n)}^{\overline{\alpha}}((u_{2}, u_{2}))$ for all $u_{1} \in C_{i}$ and all $u_{2} \in V(G) \setminus C_{i}$. It follows that:
\[
(G[C_{i}], \chi_{2,O(\log n)}^{\overline{\alpha}}) \not \cong (H[C_{i}'], \chi_{2,O(\log n)}^{\overline{\alpha'}}).
\]

\noindent If we do not have both $v \in C_{i}$ and $v' \in C_{i}'$, then Spoiler pebbles $w \mapsto f'(w)$. Define:
\[
\overline{\alpha}_{1} := \begin{cases}
\overline{\alpha} & : v \in C_{i} \text{ and } v' \in C_{i}', \\
(\overline{\alpha}, w) & : \text{otherwise}.
\end{cases}
\]

\noindent Define $\overline{\alpha'_{1}}$ analogously. Observe that:
\[
(G[C_{i}], \chi_{2,O(\log n)}^{\overline{\alpha}_{1}}) \not \cong (H[C_{i}'], \chi_{2,O(\log n)}^{\overline{\alpha'}_{1}}).
\]

\noindent We again consider the flip extensions $\overline{s_{1}}, \overline{s_{1}'}$. We have that $C_{i}$ forms a connected component in $G^{\overline{s_{1}}}$, and similarly $C_{i}'$ forms a connected component in $H^{\overline{s_{1}}}$. Thus, we may remove pebbles from outside of $C_{i}$ (respectively $C_{i}'$) without changing whether $\chi_{2,O(\log n)}$ distinguishes $G[C_{i}]$ and $H[C_{i}']$. So we may remove all pebbles $\overline{a}, \overline{b}, \overline{a}_{2}, \overline{b}_{2}$. If $w \mapsto f'(w)$ was additionally pebbled, we may also remove the pebble pair $v \mapsto v'$. For convenience, we label:
\[
z = \begin{cases} v & : v \in C_{i} \text{ and } v' \in C_{i}', \\ 
w & : \text{otherwise}.
\end{cases}
\]

\noindent Define $z'$ analogously. So now we have that:
\[
(G[C_{i}], \chi_{2,O(\log n)}^{(\overline{a}_{1}, \overline{b}_{1}, z)}) \not \cong (H[C_{i}'], \chi_{2,O(\log n)}^{(\overline{a'_{1}}, \overline{b'_{1}}, z')});
\]

\noindent otherwise, by \Cor{cor:GN3.12}, Spoiler can win with $2$ pebbles (reusing pebbles we have removed) and $O(\log n)$ additional rounds. Thus, we have not used any additional pebbles in this case. We now apply the inductive hypothesis to $t_{1}$ to deduce that Spoiler wins from $((\overline{a}_{1}, \overline{b}_{1}, z), (\overline{a'_{1}}, \overline{b'_{1}}, z'))$.

\item \textbf{Case 2:} Suppose $C \cap C_{i} \subseteq X_{2}$. As $X_{1}$ is defined with respect to the flip extension $\overline{s_1}$, this case is not symmetric with respect to Case 1. Define $M := C \cap C_{i}$ and $M' := C' \cap C_{i}'$.

Spoiler begins by pebbling $w \mapsto f'(w)$ (recall that $w$ was defined right above Equation~(\ref{eq:BeforeCase1})). Now consider the flip extension $\overline{s_{2}}$ (recall that $\overline{s_2}, \overline{s_{2}'}$ were defined above). As the configuration $\overline{\alpha}$ has been pebbled and by Lemma~\ref{lem:Connectivity}, we have that $\chi_{2,O(\log n)}^{\overline{\alpha}}((u_{1}, u_{1})) \neq \chi_{2,O(\log n)}^{\overline{\alpha}}((u_{2}, u_{2}))$ for all $u_{1} \in M$ and all $u_{2} \in V(G) \setminus M$.

Let $f'' : V(G) \to V(H)$ be the bijection that Duplicator selects at the next round. As $\overline{\alpha}$ has been pebbled, we may assume that $f''(M) = M'$; otherwise, Spoiler pebbles some element of $M$ that does not map to an element of $M'$. Then by \Lem{lem:Connectivity}, Spoiler wins with $1$ pebble (removing a pebble of $\overline{\alpha}_{1}$) and $O(\log n)$ additional rounds. Let $\{ D_{1}, \ldots, D_{q}\}$ be the  components of $\text{Comp}(G, \overline{s}_{2})$ that have non-empty intersection with $M$, and define $\{ D_{1}', \ldots, D_{q}'\}$ to be the components of $\text{Comp}(H, \overline{s_{2}'})$ that have non-empty intersection with $M'$.

By Lemma~\ref{lem:PreservesIsomorphism} and as $\overline{\alpha}, w$ have been pebbled, we have that:
\[
(G^{\overline{s_{2}}}[M], \chi_{2,O(\log n)}^{(\overline{\alpha}, w)}) \not \cong (H^{\overline{s_{2}'}}[M'], \chi_{2,O(\log n)}^{(\overline{\alpha}, w)}).
\]

\noindent By \Lem{Observation4.3}, there exists some $z \in M$ such that:
\[
(G^{\overline{s_{2}}}[M \cap D_{j}], \chi_{2,O(\log n)}^{(\overline{\alpha}, w)}) \not \cong (H^{\overline{s_{2}'}}[M' \cap D_{j}'], \chi_{2,O(\log n)}^{(\overline{\alpha}, w)}),
\]

\noindent where $j \in [q]$ is the unique component containing $z$ and $D_{j}'$ is the corresponding component containing $f''(z)$. It follows that:
\[
(G^{\overline{s_{2}}}[M \cap D_{j}], \chi_{2,O(\log n)}^{(\overline{\alpha}, w, z)}) \not \cong (H^{\overline{s_{2}'}}[M' \cap D_{j}'], \chi_{2,O(\log n)}^{(\overline{\alpha}, w, f''(z))}).
\]

\noindent Applying Lemma~\ref{lem:PreservesIsomorphism} again, we have that:
\[
(G[M \cap D_{j}], \chi_{2,O(\log n)}^{(\overline{\alpha},w, z)}) \not \cong (H[M' \cap D_{j}'], \chi_{2,O(\log n)}^{(\overline{\alpha}, w, f''(z))}).
\]

\noindent Now if $w \in D_{i}$ and $w' \in D_{i}'$, Spoiler does not place an additional pebble. Otherwise, Spoiler pebbles $z \mapsto f''(z)$. For convenience, we define:
\[
x = \begin{cases} w & : w \in D_{j} \text{ and } w' \in D_{j}' \\
z & : \text{otherwise.}
\end{cases}
\]

\noindent Define $x'$ analogously. Now as
\[
(G[M \cap D_{j}], \chi_{2,O(\log n)}^{(\overline{\alpha},w, z)}) \not \cong (H[M' \cap D_{j}'], \chi_{2,O(\log n)}^{(\overline{\alpha}, w, f''(z))}),
\]

\noindent we have that:
\[
(G[M \cap D_{j}], \chi_{2,O(\log n)}^{(\overline{\alpha},w, x)}) \not \cong (H[M' \cap D_{j}'], \chi_{2,O(\log n)}^{(\overline{\alpha}, w, x)}).
\]

\noindent Now in the flipped graph $G^{\overline{s_{2}}}$, $D_{i}$ forms a connected component. Similarly, in $H^{\overline{s_{2}'}}$,  $D_{i}'$ forms a connected component. So removing any pebbles from outside of $D_{i}$ (respectively, $D_{i}'$) does not affect whether $\chi_{2,O(\log n)}$ distinguishes $D_{i}$ from $V(G) \setminus D_{i}$ (respectively, whether $\chi_{2,O(\log n)}$ distinguishes $D_{i}'$ from $V(H) \setminus D_{i}'$). So we may remove all pebbles, so that only $(\overline{a}_{2}, \overline{b}_{2}, x) \mapsto (\overline{a'_{2}}, \overline{b'_{2}}, x')$ remain pebbled and still obtain that:
\[
(G[D_{j}], \chi_{2,O(\log n)}^{(\overline{a}_{2}, \overline{b}_{2} x)}) \not \cong (H[D_{j}'], \chi_{2,O(\log n)}^{(\overline{a'_{2}}, \overline{b'_{2}}, x')}).
\]

\noindent Otherwise, Spoiler wins using $2$ pebbles we have removed and $O(\log n)$ additional rounds. We now apply the inductive hypothesis to $t_{2}$, to deduce that Spoiler wins from $((\overline{a}_{2}, \overline{b}_{2}, z), (\overline{a'_{2}}, \overline{b'_{2}}, z'))$.
\end{itemize}

\noindent \\ So by induction, Spoiler has a winning strategy. It remains to analyze the round and pebble complexities. We first claim that only $O(\log n)$ rounds are necessary. At each node of the host tree, only a constant number of rounds are necessary unless (i) Duplicator selects a bijection that does not respect connected components, or (ii) we apply Corollary~\ref{cor:GN3.12}. Note that either case can only happen once. If Duplicator selects a bijection that does not respect connected components, then by \Lem{lem:Connectivity}, Spoiler wins with $O(\log n)$ rounds. Thus, as the height of $T$ is $O(\log n)$, only $O(\log n)$ rounds are necessary.

We will now analyze the number of pebbles, following the same strategy as Grohe \& Neuen \cite{grohe2019canonisation}. We can pebble $(\overline{a}, \overline{b}, \overline{a}_{1}, \overline{b}_{1}, \overline{a}_{2}, \overline{b}_{2})$ using $12k$ pebbles, as $(T, \gamma)$ has width at most $2k$. As $\overline{a} \subseteq \overline{a}_{1} \cup \overline{a}_{2}$, we need not pebble $\overline{a}$ and so can use only $10k$ pebbles. By \Lem{lem:ExistNice}, Spoiler can choose nice split pairs $(\overline{a}, \overline{b})$ and $(\overline{a}_{i}, \overline{b}_{i})$ such that additionally $\overline{b}_{i} \cap \overline{X} \subseteq \overline{b}$. So $\overline{b}_{i} \subseteq \overline{b} \cup \overline{a}_{3-i}$. This brings us down to $6k$ pebbles.

At most two of $x, v, w$ are pebbled at a given round. By \Lem{lem:ExistNice}, we can remove pebbles from $\overline{b}_{2}$ in Case 1 or $\overline{b}_{1}$ in Case 2. So only one additional pebble is necessary. Furthermore, if Duplicator selects bijections that do not respect connected components corresponding to pebbled vertices, then Duplicator can remove pebbles from $\overline{b}_{2}$ in Case 1 or $\overline{b}_{1}$ in Case 2 to win (see \Lem{lem:Connectivity}). So at most $1$ additional pebble is required. Thus, Spoiler has a winning strategy with $6k+3$ pebbles on the board and $O(\log n)$ rounds. The result follows.
\end{proof}

\begin{remark}
If we use a rank decomposition of width $k$, then we are able to obtain a slight improvement upon \cite[Theorem~4.1]{grohe2019canonisation}, using $(3k+3)$-WL (without controlling for rounds) rather than $(3k+4)$-WL. 
\end{remark}

\begin{corollary}
The Weisfeiler--Leman dimension of graphs of rank-width $k$ is  $\leq 3k+3$.
\end{corollary}

\section{Canonical Forms in Parallel}

In this section, we will establish the following.

\begin{theorem} \label{thm:CanonicalForms}
Let $G$ be a graph on $n$ vertices, of rank-width $k$. We can compute a canonical labeling for $G$ using a $\textsf{TC}$ circuit of depth $O(\log^{2} n)$ and size $n^{O(16^k)}$.
\end{theorem}

For $k \in O(1)$, this yields an upper bound of $\textsf{TC}^{2}$. \\

\noindent We will prove Theorem~\ref{thm:CanonicalForms} via the individualization-and-refinement paradigm. Our strategy is similar to that of Köbler \& Verbitsky \cite{KoblerVerbitsky}, who established the analogous result for treewidth. We will begin by briefly recalling their approach. K\"obler \& Verbitsky began by enumerating ordered sequences of vertices of length $\leq k+1$, testing whether each such sequence disconnected the graph. In particular, Köbler \& Verbitsky crucially used the fact that a graph of treewidth $k$ admits a so-called \textit{balanced} separator $S$ of size $\leq k+1$, which splits $G$ into connected components each of size $\leq n/2$. Köbler \& Verbitsky then colored the vertices of each connected component of $G-S$ according to how they connected back to $S$. As graphs of bounded treewidth are hereditary (closed under taking induced subgraphs), Köbler \& Verbitsky were then able to recurse on the connected components. The existence of balanced separators guarantees that only $O(\log n)$ such recursive calls are needed.

Instead of relying on balanced separators, it is sufficient to guarantee that after $O(\log n)$ recursive calls, each connected component will be a singleton. To this end, we again leverage the result of \cite{CourcelleKante2007}, who showed that a graph of rank-width $k$ admits a rank decomposition $(T, \gamma)$ of width $\leq 2k$ and height $O(\log n)$.   

Thus, we would intuitively like to descend along such a rank decomposition $(T, \gamma)$ of width $\leq 2k$ and height $O(\log n)$. Fix a node $t \in V(T)$, and let $t_{1}$ be the left child and $t_{2}$ be the right child of $t$. We would then enumerate over all pairs of flip extensions $( (\overline{a_{1}}, \overline{b_{1}}, f_{1}), (\overline{a_{2}}, \overline{b_{2}}, f_{2}))$, where intuitively $\overline{s_{i}} := (\overline{a_{i}}, \overline{b_{i}}, f_{i})$ is a flip extension for $\gamma(t_{i})$. Then for each $i = 1,2$ and each component $C_{i} \in \text{Comp}(G[\gamma(t)], \overline{s_{i}})$, we apply the construction recursively. Note that we are not able to efficiently compute a rank decomposition of width $\leq 2k$ and height $O(\log n)$. Nonetheless, Lemma~\ref{lem:GNLem5.6} guarantees the existence of flip extensions that witness the decomposition of a fixed rank decomposition $(T, \gamma)$. Following an idea of Wagner \cite{WagnerBoundedTreewidth}, we  consider all possible flip extensions in parallel, and thus ensure that the flip extension which respects a fixed rank decomposition is considered by the algorithm. As we will show in Lemma~\ref{lem:ReturnsLabeling}, the existence of a rank decomposition of height $O(\log n)$ allows us to guarantee that at least one of the flip extensions considered by the algorithm will produce a labeling, and Lemma~\ref{lem:CanonicalLabeling} will then guarantee that the minimum such labeling (which is the labeling the algorithm will return) is in fact canonical. Now to the details.

We first show that we can enumerate the split pairs in a canonical manner. To this end, we will need the following lemma, which is essentially well-known amongst those working on the Weisfeiler--Leman algorithm (cf., \cite{ImmermanLander1990, grohe2019canonisation}).

Let $G$ be a graph. The $(k,r)$-Weisfeiler-Leman algorithm \textit{determines orbits of $\ell$-tuples} if, for every graph $H$, every $v \in V(G)^{\ell}$ and every $w \in V(H)^{\ell}$ such that $\chi_{k,r}(v) = \chi_{k,r}(w)$, there is an isomorphism $\varphi:V(G)\to V(H)$ such that $\varphi(v) = w$.
\begin{lemma} \label{lem:Orbits}
Let $\mathcal{C}$ be a class of graphs such that $(k,r)$-WL identifies all (colored) $G \in \mathcal{C}$. Then for any $\ell \geq 1$ and all (colored) $G \in \mathcal{C}$, $(k+\ell,r)$-WL determines the orbits of all $\ell$-tuples of vertices in $G$.
\end{lemma}

\begin{proof}
Let $G \in \mathcal{C}$ be an arbitrary (colored) graph, and let $H$ be an arbitrary graph. Let $\bar{v} \in V(G)^{\ell}, \bar{w} \in V(H)^{\ell}$. Suppose that the coloring $\chi_{k,r}$ computed by $(k+\ell,r)$-WL fails to distinguish $\bar{v}, \bar{w}$. Then $(k,r)$-WL fails to distinguish the colored graphs $(G, \chi_{k,r}^{\bar{v}})$ and $(H, \chi_{k,r}^{\bar{w}})$. As $(k,r)$-WL identifies all (colored) graphs in $\mathcal{C}$, we have that $(G, \chi_{k,r}^{\bar{v}}) \cong (H, \chi_{k,r}^{\bar{w}})$. So there is an isomorphism $\varphi : G \cong H$ mapping $\varphi(\bar{v}) = \bar{w}$. The result now follows.
\end{proof}

By Theorem~\ref{thm:MainParallel1}, we have that $(6k+3, O(\log n))$-WL identifies all graphs of rank-width $k$. As we will need to enumerate split pairs, which have length $\leq 4k$, we will run $(10k+3, O(\log n))$-WL at each stage. Lemma~\ref{lem:Orbits} ensures that enumerating the split pairs in color class order is canonical. Note that a flip function is represented as a tuple in $\{0,\ldots, n\}^{2^{4k}}$. So for a fixed split pair $(\overline{a}, \overline{b})$, we can canonically enumerate the flip functions in lexicographic order. Thus, flip extensions can be enumerated in a canonical order.

\begin{remark} \label{rmk:FlippedGraphs}
Now let $(\overline{a}, \overline{b})$ be a split pair on $G$ and $(\overline{c}, \overline{d})$ be a split pair on $H$ such that $\chi_{10k+3,O(\log n)}((\overline{a}, \overline{b})) = \chi_{10k+3,O(\log n)}((\overline{c}, \overline{d}))$. Let $f$ be a given flip function, and let $\overline{s} = (\overline{a}, \overline{b}, f), \overline{s'} = (\overline{c}, \overline{d}, f)$ be flip extensions. By Lemma~\ref{lem:Orbits}, there is an isomorphism mapping $(\overline{a}, \overline{b}) \mapsto (\overline{c}, \overline{d})$. Hence, Lemma~\ref{lem:PreservesIsomorphism} provides that the flipped graphs $G^{\overline{s}}, H^{\overline{s'}}$ are isomorphic whenever $G \cong H$. In particular, if there is an isomorphism $\varphi : G \cong H$ mapping $(\overline{a}, \overline{b}) \mapsto (\overline{c}, \overline{d})$, then $\varphi$ is also an isomorphism of $G^{\overline{s}}$ and $H^{\overline{s}'}$.
\end{remark}

\begin{lemma} \label{lem:WriteDownFlippedGraph}
Let $G$ be a graph, $X \subseteq V(G)$, and $\overline{s} = (\overline{a}, \overline{b}, f)$ be a flip extension for $G[X]$. We may write down the flipped graph $G^{\overline{s}}$ and identify the connected components of $G[X]^{\overline{s}}$ in $\textsf{L}$.
\end{lemma}

\begin{proof}
We note that $(\overline{a}, \overline{b})$ has $O(k)$ vertices, we can compute for any vertex $v$, $N(v) \cap (\overline{a} \cup \overline{b})$ in $\textsf{AC}^{0}$. Thus, we may in $\textsf{AC}^{0}$ compute for any vertex $w$, the equivalence class $[w]_{\approx_{\overline{s}}}$. Furthermore, we may compute $|N(v) \cap [w]_{\approx_{\overline{s}}}|$ in $\textsf{TC}^{0}$. Identifying the connected components of a graph is known to be $\textsf{L}$-computable \cite{Reingold}. The result now follows.
\end{proof}

We will now pause to outline the procedure for the reader. Let $\overline{s} := (\overline{a}, \overline{b}, f)$ be a flip extension for $V(G)$. We will first individualize $(\overline{a}, \overline{b})$ and apply $(10k+3, O(\log n))$-WL to $G$. For each component $C \in \text{Comp}(G, \overline{s})$, this will encode the isomorphism class of $G[C]$ (as $(6k+3, O(\log n))$-WL identifies all graphs of rank-width $\leq k$-- see Theorem~\ref{thm:WLRankWidth}), as well as how $G[C]$ connects back to the rest of $G$. It is easy to see that for any two vertices $v, w$, if $v, w$ receive the same color under $\chi_{10k+3, O(\log n)}^{(\overline{a}, \overline{b})}$, then the following conditions hold:
\begin{enumerate}[label=(\alph*)]
    \item $N(v) \cap (\overline{a} \cup \overline{b}) = N(w) \cap (\overline{a} \cup \overline{b})$, and
    \item For any vertex $u$, $|N(v) \cap [u]_{\approx \overline{s}}| = |N(w) \cap [u]_{\approx \overline{s}}|$.
\end{enumerate}

\noindent Intuitively, this coloring encodes how each given vertex connects to the rest of $G$. Precisely, let $G \cong H$ be graphs of rank-width $\leq k$, and suppose that the algorithm returns the labeling $\lambda : V(G) \to [n]$ for $G$ and labeling $\kappa : V(H) \to [n]$ for $H$ (where $n = |G| = |H|$). If $v, w \in V(G)$ belong to different components of $\text{Comp}(G, \overline{s})$, then we need to ensure that $\{v,w\} \in E(G)$ if and only if $\{(\kappa^{-1} \circ \lambda)(v), (\kappa^{-1} \circ \lambda)(w)\}\in E(H)$. By the definition of the flipped graph (see Section~\ref{sec:WLRankWidth}), conditions (a) and (b) determine precisely whether  $\{v,w\} \in E(G)$.

By Lemma~\ref{lem:WriteDownFlippedGraph}, we may write down the connected components for the flipped graph $G^{\overline{s}}$ in $\textsf{L}$. We will then sort these connected components in lexicographic order by color class, which is $\textsf{L}$-computable. It may be the case that for two connected components $C_{i}, C_{j} \in \text{Comp}(G, \overline{s})$, $G[C_{i}]$ and $G[C_{j}]$ are isomorphic and connect to the rest of $G$ in the same way, and so receive the same multiset of colors. In this case, we may arbitrarily choose whether $G[C_{i}]$ will be sorted before $G[C_{j}]$. The output will not depend on this particular choice, as there is an automorphism of $G$ which exchanges the two components. Now for each $C \in \text{Comp}(G, \overline{s})$, we will apply the procedure recursively on $G[C]$, incrementing the local depth variable by $1$. If for each connected component of $\text{Comp}(G, \overline{s})$ we are given a valid labeling, we may recover a labeling for $G$ as follows. Let $C_{j} \in \text{Comp}(G, \overline{s})$, with the labeling function $\ell_{j} : V(C_{j}) \to \{ 1, \ldots, |V(C_{j})| \}$ returned by applying our canonization procedure recursively to $G[C_{j}]$. Let $h_{j} := |C_{1}| + \cdots + |C_{j-1}|$. We will recover a canonical labeling $\ell : V(G) \to [n]$ by, for each such $j$ and $v\in C_j$, setting $\ell(v) := \ell_{j}(v) + h_{j}$. As each vertex of $G$ appears in exactly one $C_{j}$, $\ell$ is well-defined.

We stress here again that the recursive calls to the canonization procedure track the depth to ensure that we do not make $\geq 3 \cdot (\log(n) + 1)$ recursive calls. If the depth parameter is ever larger than $ 3 \cdot (\log(n) + 1)$, then the algorithm returns $\perp$ to indicate an error. In the recombine stage of our divide and conquer procedure, if any of the labelings returned for the components of $\text{Comp}(G, \overline{s})$ are $\perp$, then the algorithm simply returns $\perp$. Thus, a priori, our algorithm may not return a labeling of the vertices. We will prove later (see Lemma~\ref{lem:ReturnsLabeling}) that our algorithm actually does return a labeling.

We now give a more precise description of our algorithm and proceed to prove its correctness. We define a canonical labeling $\text{Can}(G)$ of a graph, via a subroutine $\text{Can}(G,d)$. The subroutine $\text{Can}(G,d)$ takes an $n$-vertex graph $G$ and a depth parameter $d$, and outputs either a bijection $\lambda:V(G)\to [n]$ or a failure symbol $\bot$.

In pseudocode, our canonical labeling subroutine works as follows: 
\begin{algorithm} $\textbf{Can}(G,d)$:
\noindent \\ \textbf{Input:} A colored graph $G = (V,E, \chi)$ of rank-width $\leq k$, and a parameter $d$ for depth. 
\begin{enumerate}\label{algo main}
	\item If $d > 3 \cdot (\log(n) + 1)$, return $\bot$.\label{line excess depth}
	\item If $d \leq 3 \cdot (\log(n) + 1)$ and $|V| = 1$, return $\lambda(v) = 1$. \label{line single vertex}
	\item Otherwise, if $d \leq 3 \cdot (\log(n) + 1)$ and $|V| > 1$, do the following steps:\label{line main case}
    \item Run $(10k+3, O(\log n))$-WL on $G$. \label{line run WL initially}
	\item In parallel, enumerate all possible flip extensions $\overline{s}  =(\overline{a},\overline{b},f)$ in lexicographic order, where the order on $(\overline{a}, \overline{b})$ is considered with respect to the ordering induced by the coloring $\chi_{10k+3, O(\log n)}$ (by \cite{GroheVerbitsky}, the colors are represented by numbers, and so color class order is well-defined).\label{line enumerate flip extensions}
	\item For each flip extension, $\overline{s} = (\overline{a},\overline{b},f)$,
		\begin{enumerate}
                \item Compute the coloring $\chi_{10k+3,O(\log(n))}^{(\overline{a},\overline{b})}$ applied to $G$. \label{line run WL for each extension}
                
				\item Construct the flipped graph $G^{\overline{s}}$.\label{line construct flipped graph}
				\item Compute the set of connected components $\text{Comp}(G,\overline{s})$. If $G^{\overline{s}}$ is connected, then return $\bot$. Note that there exists a rank decomposition $(T, \gamma)$ in which for all $u, v \in V(T)$, $\gamma(u) \neq \gamma(v)$. So there exists a flip extension $\overline{s}$ that splits  $G^{\overline{s}}$ into at least two connected components. \label{line find components of flipped graph}
                
				\item Order the components $C\in \text{Comp}(G,f)$ by lexicographic ordering of the multiset of colors $\chi_{10k+3,O(\log(n))}^{\overline{a},\overline{b}}(G[C])$. Let $C_1,\dots,C_\ell$ be the components in this ordering.\label{line order components}
				\item Compute $\text{Can}(d + 1,G[C_1]),\dots,\text{Can}(d + 1,G[C_\ell])$ and let $\lambda_{\overline{s},1},\dots,\lambda_{\overline{s},\ell}$ be the resulting labelings.\label{line recursive call}
				\item If $\lambda_{\overline{s},i} = \bot$ for any $i\in [\ell]$ set $\lambda_{\overline{s}} = \bot$. Otherwise, if $\lambda_{\overline{s}, 1},\dots,\lambda_{\overline{s}, \ell}$ are the (canonical) labelings returned by the recursive calls, set 
				\[
				\lambda_{\overline{s}}(v) = \lambda_{\overline{s}, i}(v) + \sum_{j = 1}^{i - 1} |C_j|
				\] where $C_i \ni v$. \label{line recombine step}
		\end{enumerate}
	\item Return the labeling $\lambda_{\overline{s}}$ corresponding to the first flip extension $\overline{s}$ (relative to the order in which the flip extensions were enumerated) that is not $\bot$. \label{line final return}  
\end{enumerate}
\end{algorithm}

We then define the canonical labeling by setting $\text{Can}(G) = \text{Can}(G,0)$ via the subroutine. We now show that our subroutine satisfies the desired properties. 

\begin{lemma} \label{lem:ReturnsLabeling}
	If $G$ is a graph of rank-width at most $k$, the above procedure terminates and does not return $\bot$. 
\end{lemma}
\begin{proof}
	For termination, we observe that at each step the depth parameter $d$ increases and that if $d$ becomes larger than $3\log(n) + 1$, the procedure returns. Hence, the procedure must terminate.

 We will now show by induction that Algorithm~\ref{algo main} returns a labeling instead of $\bot$. Fix $(T, \gamma)$ to be a rank decomposition of $G$, of width $\leq 2k$ and height $\leq 3 \cdot (\log(n) + 1)$. Let $t \in V(T)$, and let $t_{1}, t_{2}$ be the children of $t$ in $T$. We will use Lemma~\ref{lem:GNLem5.6}, which provides that for each $t \in V(T)$, there exists a flip extension $\overline{s}$ so that for every $C \in \text{Comp}(G[\gamma(t)], \overline{s})$, there exists an $i = 1, 2$ such that $C \subseteq \gamma(t_{i})$. We will use this to show that the algorithm constructs a non-empty set of labelings for $G$. As the algorithm chooses the least such labeling (with respect to the order in which the flip extensions were enumerated-- See Algorithm~\ref{algo main}, Line 5), it follows that the algorithm in fact returns a labeling. Note that while the algorithm will not be explicitly constructing $(T, \gamma)$, the algorithm still descends along $(T, \gamma)$ in one of its parallel computations.
 
 We consider first the case when $|V(G)| = 1$. In this case, the algorithm returns $\lambda(v) = 1$, where $v \in V(G)$. Now fix a node $t \in V(T)$, and let $\gamma(t)$ be the corresponding set of vertices. Suppose that $|\gamma(t)| > 1$. Let $t_{1}, t_{2}$ be the children of $t$ in $T$. By Lemma~\ref{lem:GNLem5.6}, there exists a flip extension $\overline{s} = (\overline{a}, \overline{b}, f)$ such that for every component $C \in \text{Comp}(G[\gamma(t)], \overline{s})$, either $C \in \gamma(t_{1})$ or $C \in \gamma(t_{2})$. As we consider all flip extensions of $\gamma(t)$ in parallel, one of our parallel computations will consider $\overline{s}$. We will analyze this parallel computation.

 Prior to recursively invoking the algorithm on each $G[C]$ ($C \in \text{Comp}(G[\gamma(t)], \overline{s})$), the algorithm first sorts said components. (For the purposes of showing that the algorithm yields a (not necessarily canonical) labeling, the precise ordering does not matter. We will argue later that the ordering used by the algorithm is canonical-- see Lemma~\ref{lem:CanonicalLabeling}.) For each $C \in \text{Comp}(G[\gamma(t)], \overline{s})$, the algorithm is then applied recursively to $G[C]$. 

 Now for $i = 1, 2$, let $C_{i,1}, \ldots, C_{i,j_{i}} \in \text{Comp}(G[\gamma(t)], \overline{s})$ be precisely the components in $\gamma(t_{i})$. Observe that a flip extension on $\gamma(t_{i})$ restricts to a flip extension on an individual component $C_{i,h}$ ($h \in [j_i]$). Conversely, given flip extensions $\overline{s}_{i,h}$ ($h \in [j_i]$), the union of these flip extensions induce a flip extension $\overline{s}$ on $\gamma(t_{i})$. 

 By Lemma~\ref{lem:GNLem5.6}, there exists a flip extension $\overline{s}_{i}$ such that for every component $C' \in \text{Comp}(G[\gamma(t_{i})], \overline{s_{i}})$, $C' \in \gamma(t_{i,1})$ or $C' \in \gamma(t_{i,2})$. Suppose that $s_{i}$ is the union of the flip extensions $(\overline{s}_{i,h})_{h \in [j_{i}]}$.  As, for each $h \in [j_{i}]$, the recursive call of the algorithm applied to $C_{i,h}$ will consider all flip extensions of $C_{i,h}$ in parallel. Thus, via the recursive calls to the components $C_{i,h}$ ($h \in [j_{i}]$), the algorithm will consider all flip extensions of $\gamma(t_{i})$, including the flip extension $\overline{s_{i}}$. Thus, some parallel choice will descend along $(T, \gamma)$, and so we may assume that the algorithm computes a labeling for each $C \in G([\gamma(t)], \overline{s})$. As these components are disjoint and listed in a fixed order, the algorithm in fact computes a labeling for $\gamma(t)$. 

 The result now follows by induction.
\end{proof}

\begin{lemma} \label{lem:CanonicalLabeling}
	Let $G$ be a colored graph of rank-width at most $k$ and let $H$ be an arbitrary graph. If $\lambda : V(G) \to [n]$ and $\kappa : V(H) \to [n]$ are the labelings output by Algorithm \ref{algo main} on $G$ and $H$ respectively, then $G\cong H$ if and only if the map $\kappa^{-1} \circ \lambda$ is an isomorphism. 
\end{lemma}
\begin{proof}
If $\kappa^{-1} \circ \lambda$ is an isomorphism, then clearly $G \cong H$. Thus, it suffices to show that if $G \cong H$, then $\kappa^{-1} \circ \lambda$ is an isomorphism.

The proof is by induction on the number of vertices in $G$. Assume that $|V(G)| = 1$ and $G\cong H$. Note that the algorithm returns $\lambda = \kappa = \text{id}$, the identity permutation, on a graph with one vertex. Thus, $\kappa^{-1} \circ \lambda$ is an isomorphism as desired.

Now suppose that $|V(G)| > 1$. Let $\lambda$ be the labeling returned for $G$ and $\kappa$ the labeling returned for $H$ (recall that by  Lemma~\ref{lem:ReturnsLabeling}, and the fact that $G$ has rank-width at most $k$, we have that $\lambda \neq \bot$ and thus $\kappa \neq \bot$). Let $(\overline{a}, \overline{b})$ be the split pair the algorithm selects for $\lambda$ on the initial call (that is, when the algorithm is invoked on $G$ with $\text{depth} = 0$). We note that, by the algorithm, $\chi_{10k+3, O(\log n)}( (\overline{a}, \overline{b}))$ belongs to the minimal color class where a labeling was returned. Let $(\overline{a}', \overline{b}')$ be the corresponding split pair of $H$ selected for $\kappa$. 

We now claim that $(10k +3,O(\log(n))$-WL must assign the same color to the tuples $(\overline{a},\overline{b})$ and $(\overline{a}',\overline{b}')$.  Suppose for contradiction that the tuples receive different colors, and assume without loss of generality that $(\overline{a},\overline{b})$ receives a color lower than the color assigned to $(\overline{a}',\overline{b}')$.  By Theorem~\ref{thm:WLRankWidth}, $(6k+3, O(\log n))$-WL identifies all (colored) graphs of rank-width $\leq k$. Hence there is a tuple $(\overline{u},\overline{v})$ of vertices in $H$ such that $\chi_{10k+3,O(\log(n))}((\overline{a},\overline{b})) = \chi_{10k+3,O(\log(n))}((\overline{u},\overline{v}))$. Thus, by Lemma~\ref{lem:Orbits}, there is an isomorphism $\varphi : V(G) \to V(H)$ such that $\varphi((\overline{a}, \overline{b})) = (\overline{u}, \overline{v})$. As $G \cong H$, either the algorithm returns $\bot$ when run on $H$ with the split pair $(\overline{u},\overline{v})$ or the algorithm returns a labeling $\kappa'$. In the first case, since all choices of the algorithm may be mapped back to $G$ via the automorphism $\varphi$, it would follow that the algorithm on $G$ with split pair $(\overline{a},\overline{b})$ did not return a labeling, contradicting the assumption that the algorithm returned the labeling $\lambda$ for $G$. As $(\overline{a},\overline{b})$ is the split pair responsible which produced the overall labeling of $G$, we must turn to the case where the algorithm returns a labeling on $H$ with split pair $(\overline{u},\overline{v})$. However, since $ \chi_{10k+3,O(\log(n))}((\overline{u},\overline{v})) = \chi_{10k+3,O(\log(n))}((\overline{a},\overline{b})) < \chi_{10k+3,O(\log(n))}((\overline{a}',\overline{b}'))$, the algorithm should have used $(\overline{u},\overline{v})$ for its final split pair, which contradicts the minimality of $(\overline{a}',\overline{b}')$.  Thus we can conclude that $\chi_{10k+3,O(\log(n))}((\overline{a},\overline{b})) = \chi_{10k+3,O(\log(n))}((\overline{a}',\overline{b}'))$. In particular, there is an isomorphism 
\[
\varphi : (G, \chi_{6k+3, O(\log n)}^{(\overline{a}, \overline{b})}) \cong (H, \chi_{6k+3, O(\log n)}^{(\overline{a}', \overline{b}')}).
\]

\noindent Now as the algorithm enumerates the flip extensions in lexicographical order, it thereby considers the flip functions in lexicographical order. As the ordering on flip functions does not depend on the choice of split pair, and we have that $G\cong H$, the flip function $f : (2^{\overline{a} \cup \overline{b}})^{2} \to [n] \cup \{ \bot\}$ selected for $G$ will also be used for $H$ (here, we abuse $\overline{a} \cup \overline{b}$ to denote the indices of the vertices as they appear in $(\overline{a}, \overline{b})$). Write $\overline{s} := (\overline{a}, \overline{b}, f)$ and $\overline{s}' := (\overline{a}', \overline{b}', f)$.

The algorithm next computes the flipped graphs $G^{\overline{s}}$ and $H^{\overline{s}'}$. By Lemma~\ref{lem:PreservesIsomorphism}, we have that:
\begin{align} \label{eqn:FlippedIso}
(G^{\overline{s}}, \chi_{6k+3, O(\log n)}^{(\overline{a}, \overline{b})}) \cong (H^{\overline{s}'}, \chi_{6k+3, O(\log n)}^{(\overline{a}', \overline{b}')}).
\end{align}

\noindent It follows that $|\text{Comp}(G, \overline{s})| = |\text{Comp}(H, \overline{s}')|$. Denote $\ell := |\text{Comp}(G, \overline{s})|$. Label the components of $\text{Comp}(G, \overline{s})$ as $C_{1}, \ldots, C_{\ell}$, and the components of $\text{Comp}(H, \overline{s}')$ as $D_{1}, \ldots, D_{\ell}$. Furthermore, by (\ref{eqn:FlippedIso}), there exists a bijection $\psi : [\ell] \to [\ell]$ such that for all $i \in [\ell]$, $G[C_{i}] \cong H[D_{\psi(i)}]$. In particular, as we compute $\chi_{10k+3, O(\log n)}^{(\overline{a}, \overline{b})}$ at line 6(a), the isomorphism class of $G[C_{i}] \cong H[D_{\psi(i)}]$ takes into account how $G[C_{i}]$ connects to the rest of $G$ and how $H[D_{\psi(i)}]$ connects back to the rest of $H$ (see the discussion in the two paragraphs immediately below Lemma~\ref{lem:WriteDownFlippedGraph}).

As the algorithm sorts the components of $\text{Comp}(G, \overline{s})$ (respectively, $\text{Comp}(H, \overline{s}')$), we may without loss of generality take $\psi$ to be the identity permutation.

By the inductive hypothesis, we may assume that for each $i \in [\ell]$, the algorithm computes a labeling $\ell_{i} : C_{i} \to [|C_{i}|]$, a labeling $\kappa_{i} : \psi'(C_{i})) \to [|C_{i}|]$, and that $ \kappa_{i}^{-1} \circ \ell_{i}$ is an isomorphism. Now by construction, if $v \in C_{i}$, then 
\[
\lambda(v) = \lambda_{i}(v) + \sum_{j=1}^{i-1} |C_{j}|,
\]

\noindent and $\kappa$ is defined analogously. As $C_{i} \cap C_{h} = \emptyset$ (respectively, $D_{i} \cap D_{h} = \emptyset$) whenever $i \neq h$, $\lambda$ and $\kappa$ are well-defined. Furthermore, as $\kappa_{i}^{-1} \circ \lambda_{i}$ is an isomorphism of $G[C_{i}] \cong H[D_{i}]$ for each $i \in [\ell]$, $\kappa|_{H[D_{i}]}^{-1} \circ \lambda|_{G[C_{i}]}$ is an isomorphism of $G[C_{i}] \cong H[D_{i}]$. 

Now suppose that $v, w$ belong to different components of $\text{Comp}(G, \overline{s})$. Let $v' := (\kappa^{-1} \circ \lambda)(v)$ and $w' := (\kappa^{-1} \circ \lambda)(w)$. We will show that $vw \in E(G)$ if and only if $v'w' \in E(H)$. By the definition of the flipped graph (see Section~\ref{sec:WLRankWidth}), we can determine whether $vw \in E(G)$ based on $N(v) \cap (\overline{a} \cup \overline{b})$, $N(w) \cap (\overline{a} \cup \overline{b})$, and $|N(v) \cap [w]_{\equiv s}|$. All of this information is encoded in $\chi_{10k+3, O(\log n)}^{(\overline{a}, \overline{b})}((v,w))$, and $\chi_{10k+3, O(\log n)}^{(\overline{a}', \overline{b}')}((v,w))$. Thus, $vw \in E(G)$ if and only if $v'w' \in E(H)$. 

It follows that the map $\kappa^{-1} \circ \lambda$ is an isomorphism. The result now follows by induction.
\end{proof}

We now prove Theorem~\ref{thm:CanonicalForms}.
\begin{proof}[Proof of Theorem~\ref{thm:CanonicalForms}]
    Let $G$ be a graph with $n$ vertices and rank-width at most $k$, and let $\lambda:V(G)\to [n]$ be the output of $\text{Can}(G,3\log(n) + 1)$. By Lemma \ref{lem:ReturnsLabeling}, $\lambda$ is not $\bot$. By Lemma~\ref{lem:CanonicalLabeling}, $\lambda$ is a canonical labeling of $G$. Thus it remains to determine the complexity of the algorithm.

    At each recursive call to $\text{Can}(G,d)$, we invoke $(10k+3, O(\log n))$-WL on $G$, once at line (\ref{line run WL initially}), and then in parallel for each flip extension. By the parallel WL implementation due to Grohe \& Verbitsky \cite{GroheVerbitsky}, our calls to $(10k+3, O(\log n))$-WL are $\textsf{TC}^{1}$-computable. By Lemma~\ref{lem:WriteDownFlippedGraph}, we may write down the flipped graph $G^{\overline{s}}$ and identify its connected components in $\textsf{L}$. Thus, the non-recursive work within a single call to $\text{Can}(G,d)$ is $\textsf{TC}^{1}$-computable. Now $\text{Can}(G,d)$ makes $n^{O(16^{k})}$ recursive calls. The height of our recursion tree is $O(\log n)$. Thus, our circuit has size $n^{O(16^{k})}$, as desired. In particular, for fixed $k$, our algorithm is $\textsf{TC}^{2}$-computable.
\end{proof}

\section{Logarithmic Weisfeiler--Leman and Treewidth}

In the process of our work, we came across a way to modestly improve the descriptive complexity for graphs of bounded treewidth. Our main result in this section is the following.

\begin{theorem} \label{thm:Treewidth}
The $(3k+6)$-dimensional Weisfeiler--Leman algorithm identifies graphs of treewidth $k$ in $O(\log n)$ rounds.
\end{theorem}

In order to prove \Thm{thm:Treewidth}, we utilize a result of~\cite{Bodlaender} that graphs of treewidth $k$ admit a \textit{binary} tree decomposition of width $\leq 3k+2$ and height $O(\log n)$. With this decomposition in hand, we leverage a pebbling strategy that is considerably simpler than that of Grohe \& Verbitsky~\cite{GroheVerbitsky}.

\begin{remark}
Grohe \& Verbitsky~\cite{GroheVerbitsky} previously showed that the $(4k+3)$-WL identifies graphs of treewidth $k$ in $O(\log n)$ rounds. As a consequence, they obtained the first $\textsf{TC}^{1}$ isomorphism test for this family. In light of the close connections between Weisfeiler--Leman and $\textsf{FO} + \textsf{C}$~\cite{ImmermanLander1990, CFI}, they also obtained that if $G$ has treewidth $k$, then there exists a $(4k+4)$-variable formula $\varphi$ in $\textsf{FO} + \textsf{C}$ with quantifier depth $O(\log n)$ such that whenever $H \not \cong G$, $G \models \varphi$ and $H \not \models \varphi$. In light of \Thm{thm:Treewidth}, we obtain the following improvement in the descriptive complexity for graphs of bounded treewidth.
\end{remark}

\begin{corollary} \label{cor:CorMain}
Let $G$ be a graph of treewidth $k$. Then there exists a formula $\varphi_{G} \in \mathcal{C}_{3k+7, O(\log n)}$ that identifies $G$ up to isomorphism. That is, for any $H \not \cong G$, $G \models \varphi_{G}$ and $H \not \models \varphi_{G}$.
\end{corollary}

We begin by introducing some useful lemmas.

\begin{lemma} \label{lem:separators}
Let $G, H$ be graphs. Suppose that a separator $S \subseteq V(G)$ has been pebbled. If the corresponding pebbled set $S' \subseteq V(H)$ is not a separator of $H$, then Spoiler can win with $3$ additional pebbles and $O(\log n)$ additional rounds.    
\end{lemma}

\begin{proof}
Spoiler begins by pebbling vertices $v, w$ in two different components of $G - S$. Let $(v', w')$ be the corresponding pebbled vertices in $V(H)$. As $S'$ is not a separator of $H$, there is a $v'-w'$ path in $H - S$. However, there is no $v-w$ path in $G - S$. The argument now follows  identically as in the proof of Lemma~\ref{lem:Connectivity}.     
\end{proof}

Let $G$ be a connected graph, and let $(T, \beta)$ be a tree decomposition of $G$. This next lemma states that if we have pebbled the vertices of some node $\beta(t)$, then Spoiler can force Duplicator to preserve a given subtree $T'$ (setwise) of the tree decomposition by pebbling some vertex $v \in V(G)$ where there exists $u \in V(T')$ such that $v \in \beta(u)$.

\begin{lemma} \label{lem:children}
Let $G$ be a connected graph, and let $(T, \beta)$ be the binary tree decomposition of $G$ afforded by \cite{Bodlaender}. Let $t \in V(T)$, and suppose that each vertex in $\beta(t)$ has been pebbled. Let $C$ be the connected component of $T - tu$ that contains $u$, and let $T' := C \cup tu$.

Let $v,w \in V(G)$ be vertices contained in the subgraph of $G$ induced by $T'$, such that $v,w \not \in \beta(t)$. Suppose that $(v,w) \mapsto (v', w')$ are pebbled. Let $f : V(G) \to V(H)$ be Duplicator's bijection. If $v', w'$ belong to different components of $H \setminus f( \beta(t) \setminus \beta(u))$, then Spoiler can win with $1$ additional pebble and $O(\log n)$ additional rounds.
\end{lemma}

\begin{proof}
This follows immediately from Lemma~\ref{lem:Connectivity} and Lemma~\ref{lem:separators}.    
\end{proof}

\begin{proof}[Proof of Theorem~\ref{thm:Treewidth}]
Let $G$ be a graph of treewidth at most $k$, and $H$ a graph not isomorphic to $G$. Without loss of generality, we may assume that $G$ is connected. Otherwise, Duplicator selects a bijection $f : V(G) \to V(H)$ mapping some vertex $v \in V(G)$ to a vertex $f(v)$, where the component of $G$ containing $v$ is not isomorphic to the component containing $f(v)$. Spoiler begins by pebbling $v$. We will be able to reuse this pebble later, and so it will not add to our count.

Let $(T, \beta)$ be a tree decomposition for $G$ of width $\leq 3k+2$ and height $O(\log n)$, with $T$ a binary tree, as prescribed by~\cite{Bodlaender}. Let $s$ be the root node of $T$. Spoiler begins by pebbling the vertices of $\beta(s)$, using $\leq 3k+3$ pebbles. Let $f : V(G) \to V(H)$ be Duplicator's bijection. If $G[\beta(s)] \not \cong H[f(\beta(s))]$, then Spoiler wins. So suppose that $G[\beta(s)] \cong H[f(\beta(s))]$. 

Let $\ell$ be the left child of $s$, and $r$ be the right child of $s$ in $T$. Denote $S_{\ell} := \beta(s) \cap \beta(\ell)$ and $S_{r} := \beta(s) \cap \beta(r)$. Now as $(T, \beta)$ is a tree decomposition, we have that for any edge $uv \in E(T)$, $\beta(u) \cap \beta(v)$ is a separator of $G$. By Lemma~\ref{lem:separators}, we may assume that Spoiler selects bijections $f : V(G) \to V(H)$ that preserve separators of pebbled vertices, or Spoiler wins with $3$ additional pebbles (for a total of $3k+6$ pebbles) and $O(\log n)$ additional rounds. So we may assume that $f(S_{\ell}), f(S_{r})$ are separators of $H$.

Using two additional pebbles (for a total of $3k+5$ pebbles), Spoiler pebbles vertices $v_{\ell}, v_{r} \in V(G)$ belonging to the left and right subtrees of $T$ from $s$ respectively (that is, there exist nodes $t_{1}, t_{2} \in V(T)$ such that $t_{1}$ is in the left subtree from $s$ and $v_{\ell} \in \beta(t_{1})$, and $t_{2}$ is in the right subtree from $s$ and $v_{r} \in \beta(t_{2})$). Let $C_{\ell}$ be the connected subgraph of $G$ induced by the left subtree of $T$ (where for clarity, $C_{\ell}$ includes $S_{\ell}$), and define $C_{r}$ analogously for the right subtree of $T$. 

By Lemma~\ref{lem:children}, Duplicator must select bijections that preserve $C_{\ell}, C_{r}$ setwise (or Spoiler wins with $1$ additional pebble, for a total of $3k+6$ pebbles, and $O(\log n)$ rounds). As $S_{\ell}, S_{r}$ are separators of $G$, we have that $C_{\ell} \cap C_{r} \subseteq \beta(s)$. Furthermore, as $f(S_{\ell}), f(S_{r})$ are separators of $H$, we have that $f(C_{\ell}) \cap f(C_{r}) \subseteq f(\beta(s))$. As $G \not \cong H$ and $G[\beta(s)] \cong H[f(\beta(s))]$, we necessarily have that $C_{\ell} \not \cong f(C_{\ell})$ or $C_{r} \not \cong f(C_{r})$. Without loss of generality, suppose that $C_{\ell} \not \cong f(C_{\ell})$. Spoiler now pebbles each vertex of $\beta(\ell)$, reusing $v_{r}$ and all but one pebble of $\beta(s) \setminus \beta(\ell)$. Note that as $\beta(s) \cap \beta(\ell)$ and a single element of $\beta(s) \setminus \beta(\ell)$ have remained pebbled, we have by Lemma~\ref{lem:children} that Duplicator must select bijections that map $C_{\ell} \mapsto f(C_{\ell})$ (setwise) and $V(G) \setminus C_{\ell} \mapsto V(H) \setminus f(C_{\ell})$ (setwise), or Spoiler wins with $1$ additional pebble and $O(\log n)$ rounds.

We may thus iterate on the above argument, starting from $\ell$ as the root node in our subtree in the tree decomposition $(T, \beta)$.  As $G \not \cong H$, we will eventually reach a stage (such as when all of $\beta(t)$ is pebbled for some leaf node $t \in V(T)$) where the map induced by the pebbled vertices does not extend to an isomorphism. Note that in our iterated strategy, we may reuse the pebble in $\beta(s)$ in an application of Lemma~\ref{lem:children} or Lemma~\ref{lem:separators} applied to the left or right children of $\ell$. This ensures that $\leq 3k+6$ pebbles will be on the board at any given round.

It remains to analyze the number of rounds. Observe that we use at most $O(\log n)$ rounds at a given node of $T$ as prescribed by Lemma~\ref{lem:separators} or Lemma~\ref{lem:children}; however, invoking either lemma results in Spoiler winning. Otherwise, we use only $O(k)$ rounds at that node of $T$. As $T$ has height $O(\log n)$ and $k$ is bounded, this results in $O(\log n)$ rounds, as desired. The result now follows.
\end{proof}

\section{Conclusion}

We showed that the $(6k+3, O(\log n))$-WL algorithm identifies graphs of bounded rank-width. As a consequence, we obtain a $\textsf{TC}^{1}$ upper bound for isomorphism testing of graphs of bounded rank-width. In the process, we also improved the Weisfeiler--Leman dimension from $3k+4$ \cite{grohe2019canonisation} to $3k+3$, though it is not known if even $(3k+4)$-WL can identify graphs of bounded rank-width in $O(\log n)$ rounds. We conclude with several open questions.

It would be of interest to close the gap between the $(3k+3)$-WL bound where the iteration number is unknown and our $(6k+3)$-WL upper bound to obtain $O(\log n)$ rounds. One possible strategy would be to show that the exists a rank decomposition of width $k$ where the host tree has height $O(\log n)$.

\begin{question}
Let $G$ be a graph of rank-width $k$. Does there exist a rank decomposition $(T, \gamma)$ of width $k < c < 2k$ such that $T$ has height $O(\log n)$?
\end{question}

Courcelle \& Kant\'e \cite{CourcelleKante2007} showed that a rank decomposition of width $2k$ exists with a host tree of height $3(\log(n) + 1)$. Decreasing the width to some $k \leq c \leq 2k$ at the cost of increasing the height of the host tree by a constant factor would immediately yield improvements. More generally, in light of the correspondence between WL and $\textsf{FO} + \textsf{C}$, the width of the rank decomposition corresponds to the number of variables, and the depth of the host tree corresponds to the quantifier depth in formulas characterizing these graphs. Thus, controlling the trade-off between the width of the rank decomposition and the height of the host tree would directly translate into a trade-off between the number of variables and the quantifier depth in our logical formula. 

We note that isomorphism testing for graphs of bounded treewidth \cite{ElberfeldSchweitzer} is $\textsf{L}$-complete under $\textsf{AC}^{0}$-reductions. As graphs of bounded treewidth also have bounded rank-width, we have that isomorphism testing for graphs of bounded rank-width is $\textsf{L}$-hard under $\textsf{AC}^{0}$-reductions. We thus ask the following.

\begin{question}
Is isomorphism testing of graphs of bounded rank-width  $\textsf{L}$-complete?
\end{question}

It was already known that \textsc{GI} parameterized by rank-width was in $\textsf{XP}$ \cite{GroheSchweitzerRankWidth, grohe2019canonisation}. While our results improve the parallel complexity, they do not improve the parameterized complexity. 

\begin{question}
Does isomorphism testing of graphs parameterized by rank-width belong to $\textsf{FPT}$?
\end{question}

\bibliographystyle{alphaurl}
\bibliography{references}

\end{document}